\documentclass[reqno,12pt]{amsart}

\usepackage{amsmath,amssymb,mathtools,bbm,verbatim,enumitem,caption,float,appendix,kpfonts}
\usepackage{amsaddr}
\usepackage[T1]{fontenc}
\usepackage[backrefs]{amsrefs}
\usepackage[protrusion=true]{microtype}
\usepackage[english]{babel}
\usepackage[margin=1in]{geometry}
\usepackage[onehalfspacing]{setspace}
\usepackage[pdfusetitle,pagebackref,hypertexnames=false,bookmarks=false]{hyperref}
\numberwithin{equation}{section}							
\usepackage[nameinlink,noabbrev]{cleveref}
\expandafter\def\csname ver@etex.sty\endcsname{3000/12/31}

\usepackage{autonum}										

\allowdisplaybreaks[1]

\let\originalleft\left
\let\originalright\right
\renewcommand{\left}{\mathopen{}\mathclose\bgroup\originalleft}
\renewcommand{\right}{\aftergroup\egroup\originalright}

\makeatletter
\renewcommand*{\eqref}[1]{\hyperref[{#1}]{\textup{\tagform@{\ref*{#1}}}}}		
\makeatother

\newcommand\correspondingauthor[1]{%
  \begingroup
  \renewcommand\thefootnote{}\footnote{#1}%
  \addtocounter{footnote}{-1}%
  \endgroup
}

\newtheorem*{acknowledgment}{Acknowledgment}
\newtheorem*{data}{Data Availability}

\newtheorem{theorem}{Theorem}[section]
\newtheorem{proposition}[theorem]{Proposition}
\newtheorem{lemma}[theorem]{Lemma}
\newtheorem{corollary}[theorem]{Corollary}
\newtheorem{remark}[theorem]{Remark}
\newtheorem{definition}[theorem]{Definition}

\crefname{theorem}{Theorem}{Theorems}						
\creflabelformat{theorem}{#2{#1}#3}							
\crefname{main}{Main Theorem}{Main Theorems}				
\creflabelformat{main}{#2{#1}#3}							
\crefname{lemma}{Lemma}{Lemmas}								
\creflabelformat{lemma}{#2{#1}#3}							
\crefname{corollary}{Corollary}{Corollaries}				
\creflabelformat{corollary}{#2{#1}#3}						
\crefname{ineq}{inequality}{inequalities}					
\creflabelformat{ineq}{#2{\upshape(#1)}#3}					
\crefname{cond}{condition}{conditions}						
\creflabelformat{cond}{#2{#1}#3}							
\crefname{table}{Table}{Tables}								
\creflabelformat{table}{#2{\upshape#1}#3}					
\crefname{hypothesis}{Hypothesis}{Hypotheses}				
\creflabelformat{hypothesis}{#2{#1}#3}						
\crefname{remark}{Remark}{Remarks}							
\creflabelformat{remark}{#2{#1}#3}							

\def\Id{\mathbbm{1}}
\def\cx{\mathbb{C}}
\def\rl{\mathbb{R}}
\def\N{\mathbb{N}}

\def\Z{\mathbb{Z}}

\def\cA{\mathcal{A}}

\def\cC{\mathcal{C}}

\def\cI{\mathcal{I}}
\def\cJ{\mathcal{J}}

\def\cR{\mathcal{R}}
\def\cT{\mathcal{T}}

\def\ad{\mathrm{ad}}
\def\Ad{\mathrm{Ad}}

\def\coker{\mathrm{coker}}
\def\rd{\mathrm{d}}
\def\diag{\mathrm{diag}}

\def\dom{\mathrm{dom}}
\def\End{\mathrm{End}}

\def\rk{\mathrm{rk}}

\def\Res{\mathrm{Res}}

\def\SO{\mathrm{SO}}

\def\Spec{\mathrm{Spec}}

\def\SU{\mathrm{SU}}

\def\so{\mathfrak{so}}
\def\su{\mathfrak{su}}
\def\tr{\mathrm{tr}}
\def\rU{\mathrm{U}}

\title{Construction of Nahm data and BPS monopoles with continuous symmetries}
\date{\today}

\author{Benoit Charbonneau}
\address[Benoit Charbonneau]{Department of Pure Mathematics\\ Department of Physics and Astronomy\\ University of Waterloo, Ontario, Canada}
\email{\href{mailto:benoit@alum.mit.edu}{benoit@alum.mit.edu}}

\author{Anuk Dayaprema}
\address[Anuk Dayaprema]{Duke University}
\email{\href{mailto:anukdayaprema@gmail.com}{anukdayaprema@gmail.com}}

\author{C. J. Lang}
\address[C. J. Lang]{Department of Pure Mathematics\\ University of Waterloo}
\email{\href{mailto:cjlang@uwaterloo.ca}{cjlang@uwaterloo.ca}}

\author{\'Akos Nagy$^\star$}
\address[\'Akos Nagy]{Department of Mathematics\\ University of California, Santa Barbara}

\author{Haoyang Yu}
\address[Haoyang Yu]{Duke University}
\email{\href{mailto:haoyang.yu079@duke.edu}{haoyang.yu079@duke.edu}}

\reversemarginpar
\calclayout
\hypersetup{
	pdffitwindow	 = true,
	unicode			 = true,
	pdftoolbar		 = false,
	pdfmenubar		 = false,
	pdfstartview	 = {FitH},
	hypertexnames	 = false,
	colorlinks		 = true,
	linkcolor		 = black,
	citecolor		 = black,
	filecolor		 = black,
	urlcolor		 = black
}

\begin{document}

\correspondingauthor{$^\star$Corresponding author: \href{mailto:contact@akosnagy.com}{contact@akosnagy.com}}

\begin{abstract}
	We study solutions to Nahm's equations with continuous symmetries and, under certain (mild) hypotheses, we classify the corresponding Ans\"atze. Using our classification, we construct novel Nahm data, and prescribe methods for generating further solutions. Finally, we use these results to construct new BPS monopoles with spherical symmetry.
\end{abstract}

\maketitle

\section{Background and motivation}

When one studies a particular equation of interest and is confronted with the desire to prove the existence of solutions, the trick favored amongst all others is to impose symmetry to simplify the equation, in the hope that the reduced equation is more tractable. In gauge theory, this trick has manifested itself over and over again, and this paper is a new manifestation of it. 

\smallskip

Of all the gauge theory equations, one amongst those studied for the longest is the \emph{Bogomolny} equation expressing the relationship $F_\nabla = \ast \rd_\nabla \Phi$ between a connection $\nabla$ of curvature $F_\nabla$ over a vector bundle $E$, and an endomorphism $\Phi$ of $E$ called the \emph{Higgs field}. The study of symmetric monopoles on $\rl^3$ goes back to the very first attempts to produce monopoles back in 1974--1975, when some spherically symmetric monopoles were explicitly calculated in \cites{Prasad-Sommerfield-1975,tHooft-monopoles-1974,Polyakov-1974}. The study of axial symmetry came a few years later with \cites{Ward-2-monopole,Houston-Raiferartaigh-1981-axial-symmetric-monopole}. Since then, spherical symmetry was explored further in \cites{Bowman-onWeinberg-using-ADHMN,Dancer-NahmHyperKahler,Dancer-SU3monopoles,Prasad-1981-YMH-monopole-arbitrary-charge,Bais-Wilkinson-1979-spherical-symmetry-monopoles,Weinberg1982-continuous-family-monopole}, and cylindrical in \cites{Dancer-NahmHyperKahler,Dancer-SU3monopoles,MillerWeinberg2009-interactions-massless-monopole-clouds}. The bulk of the research on symmetric monopoles however concerned monopoles with discrete group of symmetries in \cites{BradenDAvanzoEnolski-Cyclic3-Monopoles,Braden-CyclicMonopoles-Toda-spectral,Braden-Enolski-TetrahedrallySymmetricMon,Houghton-Manton-Romao-constraintsBPS,Sutcliffe-CyclicMonopoles,Houghton-Sutcliffe-PlatonicMonopoles,HoughtonSutcliffe-Tetrahedral-cubic-monopoles,HoughtonSutcliffe-Octahedral-dodecahedral-monopoles,Sutcliffe-SW-monopoleSpectral-TodaSolitons,Sutcliffe-SymmetricMonopoles,HitchinMantonMurray-SymmetricMonopoles,Sutcliffe-moduli-tetrahedrallySymmetric-4monopoles,ORaifeartaigh-Rouhani-axiallySymmetricMonopoles}. Everything cited so far takes place on the Euclidean space, but symmetries should be useful wherever there are some, and have been exploited for studying monopoles on the hyperbolic 3-space in \cites{Manton-Sutcliffe-PlatonicHyperbolicMonopoles,CocSymmHbolicMpoles,NorRomSpectCurvesHbolicMpoles,Sutcliffe-spectral-curves-hyperbolic-from-ADHM,BolognesiCockburnSutcliffe-hyperbolicMonoples-JNRdata}, and for studying monopole chains, that is monopoles on $\rl^2 \times S^1$ in \cite{Harland2020-parabolicHiggs-cyclic-monopoles}.

\smallskip

Using symmetry is of course not just about proving existence. Because of the explicit nature of the solutions, it allows one to test other ideas. For instance, the recent work \cite{HarlandNogradi-AsymptoticTail} used the known examples to compute explicit tails of monopoles. And in \cite{HitchinMantonMurray-SymmetricMonopoles}, Hitchin, Manton, and Murray prove that the charge $k$ cyclic $\SU (2)$ monopoles form a geodesic submanifold of the total moduli space, hence this can be used to illustrate some of the monopole dynamics. 

\smallskip

The tools used to study monopoles are many, and include twistor spaces constructions, rational maps, spectral curves, Nahm transform; see \cite{Shnir-MagneticMonopoles} for a comprehensive book that also includes a lot of the physics. Our choice of tool is the Nahm transform and through it Nahm's equations. Nahm's equations were first introduced by Nahm in \cite{N82}. They form a system of ordinary differential equations on a triple of matrix-valued functions, defined on an open subset of $\rl$. Such a triple $\cT = \left( T_1, T_2, T_3 \right)$ satisfies Nahm's equations if
\begin{subequations}
\begin{align}
	\dot{T}_1	&= [T_2, T_3], \label{eq:Nahm1} \\
	\dot{T}_2	&= [T_3, T_1], \label{eq:Nahm2} \\
	\dot{T}_3	&= [T_1, T_2], \label{eq:Nahm3}
\end{align}
\end{subequations}
or, in a more compact form, involving the antisymmetric Levi-Civita $\epsilon$-tensor:
\begin{equation}
	\forall i \in \{ 1, 2, 3 \}: \quad \dot{T}_i = \frac{1}{2} \sum\limits_{j, k = 1}^3 \epsilon_{ijk} [T_j, T_k]. \label{eq:Nahmwith_LC}
\end{equation}
In practice, solving Nahm's equations is difficult as it is equivalent to a purely quadratic system of ordinary differential equations, that is a system of equations on functions $x_1 (t), \ldots, x_n (t)$ of the form
\begin{equation}
	\forall a \in \{ 1, \ldots, n \}: \quad \dot{x}_a = \sum\limits_{b, c = 1}^n C_{abc} x_b x_c,
\end{equation}
for some $C_{abc} \in \rl$. Such equations are not well-understood even in the $n = 2$ case. To circumvent these issues, we can use two things to our advantage: First of all, Nahm's equations can be formulated as a Lax pair, thus it has many conserved quantities. Second, we only search for axially and spherically symmetric solutions which further reduces the complexity of the system.

\smallskip

Nahm's equations corresponding to monopoles with maximal symmetry breaking have been understood for a long time, dating back to Nahm in \cite{N82} and then Hitchin for structure group $\SU(2)$ in \cites{H82,H83} and Hurtubise--Murray for arbitrary classical groups in \cite{HurMurMpoleConstClassGrps}. For arbitrary symmetry breaking, apart from low charge rank 2 minimal symmetry breaking work of Dancer \cite{Dancer-SU3monopoles} and the work of Houghton--Weinberg \cite{HoughtonWeinberg2002} extending it, nothing comprehensive about the Nahm transform and the behavior of the Nahm data corresponding to monopoles with arbitrary symmetry breaking was known prior to \cites{Charbonneau-VBAC,Charbonneau-CRM2017,Charbonneau-Nagy-NahmTransform}. The eigenvalues of the Higgs field at infinity correspond to singular points on the interval of definition of the solution to Nahm's equations, and at those singular points the $T_i$ have residues forming representations of $\su(2)$. Unlike in the maximal symmetry breaking case, where \cite{HurMurMpoleConstClassGrps} showed that as one approaches a singular point from the side of highest rank there are continuing and terminating components in the $T_i$ and only a continuing component from the other side, here terminating components can arise on both sides. In this paper, we do not explicitly consider the full possibilities of solutions to Nahm's equations where multiple adjoining intervals are permitted, and only produce explicit solutions on a single interval, corresponding to monopoles whose Higgs field has only two distinct eigenvalues at infinity. While our method could potentially produce spherically symmetric monopoles whose Higgs fields have more than two distinct eigenvalues at infinity, but finding such examples is beyond the scope of this paper. Thus, for our current purpose, it suffices to describe the behavior of the Nahm data at the two ends of this interval: the residues of the $T_i$ form a (possibly reducible) representation of $\su (3)$.

\smallskip

The main result of our paper, \Cref{theorem:Str_Thm_for_Spherical_Symmetry}, is a structure theorem for spherically symmetric solutions to Nahm's equations under certain reasonable conditions. Using the simplest version of this structure theorem, we provide in \Cref{theorem:long_chains} Ans\"atze for families of spherically symmetric solutions based on a long chain of representations $V_{n+2k}\oplus\cdots\oplus V_n$. We capitalize on this Ansatz for a chain of length 3 in \Cref{theorem:5+3+1Monopoles} to produce Nahm data for a spherically symmetric monopole with symmetry breaking $\mathrm{S}\left( \rU (3)\times \rU (3)\right)$. Chains of length 2 are fully explored in \Cref{theorem:3+1Monopoles,theorem:(n+2)+nMonopoles} to contrast an infinite family of solutions to Nahm's equations, indexed by $n \in \N_+$, that are defined on $(- 1, 1)$ and whose residues correspond to the representation $V_{n + 1} \oplus V_{n + 1}$ on one end and $V_2^{\oplus (n + 1)}$ on the other. The corresponding monopoles are spherically symmetric $\SU (n + 3)$ with symmetry breaking $\mathrm{S} \left( \rU (2) \times \rU (n + 1) \right)$, so their symmetry breaking is neither minimal nor maximal. 

\medskip

\subsection*{Organization of the paper}

The paper is organized as follows: In \Cref{sec:symmetries}, we recall some relevant properties of Nahm's equations, with an emphasis on the two canonical group actions, rotations and gauges. We then introduce the notions of axially and spherically symmetric solutions. In \Cref{sec:axial}, we prove an infinitesimal version of the axial symmetry condition and classify axially symmetric solutions. In \Cref{sec:spherical}, we prove an infinitesimal version of the spherical symmetry condition and we classify (under certain extra hypotheses) and construct novel spherically symmetric solutions. We also provide a method for constructing further spherically symmetric Nahm data. Finally, in \Cref{sec:Nahm_transform}, we use the spherically symmetric solutions found in \Cref{sec:spherical} to construct novel BPS monopoles using the (ADHM--)Nahm transform.

\bigskip

\section{Nahm's equations and their symmetries}
\label{sec:symmetries}

Let us begin with three well-known facts about Nahm's \cref{eq:Nahm1,eq:Nahm2,eq:Nahm3}:
\begin{enumerate}
	\item Since the function
		\begin{equation}
			M_{n \times n} (\cx)^{\oplus 3} \rightarrow M_{n \times n} (\cx)^{\oplus 3}; \quad (A, B, C) \mapsto ([B, C], [C,A], [A,B]),
		\end{equation}
		is analytic (in fact polynomial), solutions to Nahm's equations are (real) analytic and unique for any initial condition on an open interval around the initial point; see, for example, \cite{T12}*{Theorem~4.2}.

	\item Since $\su (n)$ is a matrix Lie-algebra, it is closed under commutators. Thus if a solution to Nahm's equations that is defined on a connected interval $(a, b)$ and for all $i \in \{ 1, 2, 3 \}$ and for some $t_0 \in (a, b)$ satisfies that $T_i (t_0) \in \su (n)$, then it must also satisfy that $T_i (t) \in \su (n)$ for all $t \in (a, b)$ and, again, for all $i \in \{ 1, 2, 3 \}$. In fact, $\su (n)$ can be replaced by an arbitrary Lie algebra.

	\item Any solution to Nahm's equations has the following five conserved quantities:
	\begin{subequations}
	\begin{align}
		C_1 &= \tr \left( T_2 T_3 \right), \label{eq:CC1} \\
		C_2 &= \tr \left( T_3 T_1 \right), \label{eq:CC2} \\
		C_3 &= \tr \left( T_1 T_2 \right), \label{eq:CC3} \\
		C_4 &= \tr \left( T_1^2 - T_2^2 \right), \label{eq:CC4} \\
		C_5 &= \tr \left( T_1^2 - T_3^2 \right). \label{eq:CC5}
	\end{align}
	\end{subequations}
\end{enumerate}

Let $T (\zeta) \coloneqq (T_1 + i T_2) - 2 i T_3 \zeta + (T_1 - i T_2) \zeta^2$. Solutions to Nahm's \cref{eq:Nahm1,eq:Nahm2,eq:Nahm3} have an invariant spectral curve defined by $\det (\eta - T (\zeta)) = 0$. The coefficient of $\eta^{n-2}$ in this equation contains the invariants of \cref{eq:CC1,eq:CC2,eq:CC3,eq:CC4,eq:CC5}, namely
\begin{equation}
	\tr \left( T (\zeta)^2 \right) = C_4 + 2 i C_3 + 4 (C_1 - i C_2) \zeta + (4 C_5 - 2 C_4) \zeta^2 - 4 (C_1 + i C_2) \zeta^3 + (C_4 - 2 i C_3)\zeta^4.
\end{equation} 
The moduli space of spectral curves and the accompanying spectral data (consisting of the cokernel sheaf of $\eta - T (\zeta)$ supported on the spectral curve) have been shown to be equivalent to the monopole moduli space of monopoles in the case of maximal symmetry breaking in \cite{HurMpoleClassGrps,HurMurMpoleConstClassGrps}. Other definitions of spectral data use the monopole itself, not its Nahm data. Those spectral data also determine the monopole for a larger class of monopoles; see \cite{hurtubiseMurraySpectral,MurNonAbelMpoles}.

\smallskip

Solutions to Nahm's equations typically develop singularities. In fact, in all known, nonconstant examples develop at least one, simple (that is, first order) pole. It is not known whether all nonconstant solutions develop singularities, and when they do, what type of singularities can occur. The lemma below---which is well-known in the literature and thus we state it without proof---describes the behavior of a solution at a simple pole.

For the remainder of this paper, let $\left( X_1, X_2, X_3 \right)$ the standard basis of $\so (3)$, that is, for all $i, j, k \in \{ 1, 2, 3 \}$, we have $(X_i)_{jk} = - \epsilon_{ijk}$.

\begin{lemma}
	\label{lemma:poles_and_representations}
	Let $\cT$ be a solution to Nahm's \cref{eq:Nahm1,eq:Nahm2,eq:Nahm3}, such that it is defined on an open interval $(t_0, t_1) \subseteq \rl$, and $\cT$ has a simple pole at $t_0$. Then the map
	\begin{equation}
		\so (3)	\rightarrow \su (n); \quad a_1 X_1 + a_2 X_2 + a_3 X_3 \mapsto a_1 \Res \left( T_1, t_0 \right) + a_2 \Res \left( T_2, t_0 \right) + a_3 \Res \left( T_3, t_0 \right),
	\end{equation}
	is a Lie algebra homomorphism (that is, a representation of $\so (3)$).
\end{lemma}

\smallskip

There are two canonical group actions on solutions to Nahm's equations.

\begin{definition}\label{definition:actions}
	Let $\cT = \left( T_1, T_2, T_3 \right) \in M_{n \times n} (\cx)^{\oplus 3}$.
	\begin{enumerate}
		\item For each $A \in \SO (3)$, define $^A \cT = (^A T_1, ^A T_2, ^A T_3)$ via
		\begin{equation}
			\forall i \in \{ 1, 2, 3 \}: \quad ^A T_i \coloneqq \sum\limits_{j = 1}^3 A_{ij} T_j. \label{eq:SO(3)_action}
		\end{equation}
		\item For each $U \in \SU (n)$, define $_U \cT = (_U T_1, {}_U T_2, {}_U T_3)$ via
		\begin{equation}
			\forall i \in \{ 1, 2, 3 \}: \quad _U T_i \coloneqq U T_i U^{-1}. \label{eq:SU(n)_action}
		\end{equation}
	\end{enumerate}
	We denote the induced actions of $X \in \so (3)$ and $Y \in \su (n)$ on $\cT$ as ${}^X \cT $ and $[Y, \cT]$, respectively.
\end{definition}

\smallskip

\begin{remark}
	The Nahm transform (explained in \Cref{sec:Nahm_transform}) provides a correspondence between monopoles and some solutions to Nahm's equations.	The action defined in \eqref{eq:SO(3)_action} corresponds on the monopole side to pulling back via the same rotation on $\rl^3$.
\end{remark}

\smallskip

The following proposition is straightforward, hence we state it without proof.

\begin{proposition}\label{proposition:actions}
	The actions defined by \cref{eq:SO(3)_action,eq:SU(n)_action} define smooth group actions. Furthermore, the two actions commute and preserve $\su (n)^{\oplus 3}$, that is
	\begin{equation}
		\forall A \in \SO(3): \forall U \in \SU (n): \quad ^A (_U \cT) = {}_U (^A \cT) \eqqcolon {}_U^A \cT. \label{eq:actions_commute}
	\end{equation}
	and if $\cT \in \su (n)^{\oplus 3}$, then $_U^A \cT \in \su (n)^{\oplus 3}$.
\end{proposition}

\smallskip

Let now $\cT$ be a smooth function
\begin{equation}
	\cT : (a, b) \rightarrow \su (n)^{\oplus 3}; \quad t \mapsto \cT (t) = \left( T_1 (t), T_2 (t), T_3 (t) \right).
\end{equation}
By an abuse of notation, let $_U^A \cT$ be the function defined as
\begin{equation}
	_U^A \cT : (a, b) \rightarrow \su (n)^{\oplus 3}; \quad t \mapsto {}_U^A (\cT (t)).
\end{equation}
Then we have the following lemma for solutions to Nahm's equations.

\begin{lemma}\label{lemma:invariance}
	Let $\cT$ be a solution to Nahm's \cref{eq:Nahm1,eq:Nahm2,eq:Nahm3} on some connected, open interval $I \subseteq \rl$. Then for all $A \in \SO (3)$ and for all $U \in \SU (n)$, the function $_U^A \cT$ also solves Nahm's equations.

	Furthermore, let $\widetilde{\cT}$ be another solution to Nahm's equations. If there exists $t_0 \in I \cap \dom \left( \widetilde{\cT} \right)$, $A \in \SO (3)$, and $U \in \SU (n)$, such that $\widetilde{\cT} (t_0) = {}_U^A \cT (t_0)$, then $\widetilde{\cT}$ can be extended to all of $I$, and $_U^A \cT|_{I \cap \dom \left( \widetilde{\cT} \right)} = \widetilde{\cT}|_{I \cap \dom \left( \widetilde{\cT} \right)}$.
\end{lemma}

\begin{proof}
	Since the actions of $\SO (3)$ and $\SU (n)$ commute, by \cref{eq:actions_commute}, it is enough to verify the first claim separately for rotations and gauges.

	Let $\cT$ be a solution to Nahm's \cref{eq:Nahm1,eq:Nahm2,eq:Nahm3} and $A \in \SO (3)$. Let $(a_{ij})_{i, j = 1}^3$ be the components of $A$. Recall the formula
	\begin{equation}
		\forall A \in \SO (3): \quad \sum\limits_{l, m, n = 1}^3 \epsilon_{lmn} a_{il} a_{jm} a_{kn} = \epsilon_{ijk}. \label{eq:determinant}
	\end{equation}
	Furthermore, note that for all $i \in \{ 1, 2, 3 \}$, $T_i = \sum_{l = 1}^3 a_{li} {}^AT_l$. Using these, we can compute the right hand side of \cref{eq:Nahmwith_LC} for $^A \cT$, and get
	\begin{align}
		\frac{\rd}{\rd t} (^A T)_i	&= \sum\limits_{l = 1}^3 a_{il} \dot{T}_l \\
									&= \sum\limits_{l = 1}^3 a_{il} \frac{1}{2} \sum\limits_{m, n = 1}^3 \epsilon_{lmn} [T_m, T_n] \\
									&= \frac{1}{2} \sum\limits_{i, j, k, l, m, n = 1}^3 a_{il} \epsilon_{lmn} a_{jm} a_{kn} [(^A T)_j, (^A T)_k] \\
									&= \frac{1}{2} \sum\limits_{i, j, k = 1}^3 \epsilon_{ijk} [(^A T)_j, (^A T)_k],
	\end{align}
	which completes the proof of the first claim for rotations.

	Since $[\Ad (U) (T), \Ad (U) (S)] = \Ad (U) ([T, S])$, it is clear that $\cT$ is a solution to Nahm's equations, if and only if $_U \cT$ is a solution, which completes the proof of the first claim for gauges.

	The second claim follows from the uniqueness of solutions to the initial value problem corresponding to Nahm's equations.
\end{proof}

\smallskip

The following is then immediate.

\begin{corollary}
	Let $A \in \SO (3)$, $U \in \SU (n)$, and $\cT_0 \in \su (n)^{\oplus 3}$, such that $_U^A \cT_0 = \cT_0$. If
	\begin{equation}
		\cT : (t_0 - \epsilon, t_0 + \epsilon) \rightarrow \su (n)^{\oplus 3},
	\end{equation}
	is the (unique) solution to Nahm's \cref{eq:Nahm1,eq:Nahm2,eq:Nahm3} with the initial condition $\cT (t_0) = \cT_0$, then ${}_U^A \cT = \cT$, that is
	\begin{equation}
		\forall t \in (t_0 - \epsilon, t_0 + \epsilon): \quad {}_U^A (\cT (t)) = \cT (t).
	\end{equation}
\end{corollary}

\smallskip

\begin{remark}
	If $A \in \mathrm{O} (3)$, then $^A \cT$ can still be defined via \eqref{eq:SO(3)_action}. Then \cref{eq:determinant} becomes
	\begin{equation}
		\sum\limits_{l, m, n = 1}^3 \epsilon_{lmn} a_{il} a_{jm} a_{kn} = \det (A) \epsilon_{ijk}.
	\end{equation}
	In particular, when $A \not\in \SO (3)$, the only thing that changes in the proof of \Cref{lemma:invariance} is the sign in \cref{eq:determinant}. Thus if $\cT$, is a solution to Nahm's \cref{eq:Nahm1,eq:Nahm2,eq:Nahm3}, then $(S_1, S_2, S_3) \coloneqq {}^A \cT$ is a solution to the \emph{anti-Nahm equations}:
	\begin{equation}
		\forall i \in \{ 1, 2, 3 \}: \quad \dot{S}_i = - \frac{1}{2} \sum\limits_{j, k = 1}^3 \epsilon_{ijk} [S_j, S_k].
	\end{equation}	
\end{remark}

\smallskip

\begin{definition}
	Fix $A \in \SO (3)$ and let $I \subseteq \rl$ be a nonempty, connected, and open interval. We call a solution $\cT : I \rightarrow \su (n)^{\oplus 3}$ to Nahm's \cref{eq:Nahm1,eq:Nahm2,eq:Nahm3} \emph{$A$-equivariant}, if there exists $U_A \in \SU (n)$, such that
	\begin{equation}
		^A \cT = {}_{U_A^{- 1}} \cT. \label{eq:A-equivariance}
	\end{equation}
	Similarly, if $H \subseteq \SO (3)$, then $\cT$ is \emph{$H$-equivariant}, if for all $A \in H$, $\cT$ is $A$-equivariant.
\end{definition}

\smallskip

\begin{remark}\label{rem:UA_properties}
	\Cref{eq:A-equivariance} can be written as $_{U_A}^A \cT = \cT$, thus $\cT$ is a fixed by the simultaneous actions of $A$ and $U_A$.

	Note also that when $\cT$ is $H$-equivariant (for some $H \subseteq \SO (3)$), we did not require anything else from the map
	\begin{equation}
		H \rightarrow \SU (n); \quad A \mapsto U_A,
	\end{equation}
	in particular it need not be homomorphic or smooth, and may not even be unique.

	Nonetheless, without any loss of generality, we can always assume that $U_{\Id_3} = \Id_n$ and, using \cref{eq:A-equivariance}, we can see that if $\cT$ is equivariant for both $A_1$ and $A_2$, with corresponding gauge transformations $U_{A_1}$ and $U_{A_2}$, respectively, then
	\begin{equation}
		^{A_1 A_2} \cT = {}^{A_1} ( {}^{A_2} \cT) = {}^{A_1} ( {}_{U_{A_2}^{- 1}} \cT) = {}_{U_{A_2}^{- 1}} ( {}^{A_1} \cT) = {}_{U_{A_2}^{- 1}} ( {}_{U_{A_1}^{- 1}} \cT) = {}_{U_{A_2}^{- 1} U_{A_1}^{- 1}} \cT = {}_{(U_{A_1} U_{A_2})^{- 1}} \cT,
	\end{equation}
	so $\cT$ is $A_1 A_2$-equivariant, and we can choose the corresponding gauge transformation to be $U_{A_1 A_2} = U_{A_1} U_{A_2}$. Thus the set
	\begin{equation}
		H_\cT \coloneqq \{ A \in \SO (3) | \mbox{ $\cT$ is $A$-equivariant } \},
	\end{equation}
	is a subgroup of $\SO (3)$.
\end{remark}

\smallskip

The goal of this paper is to study and construct $H$-invariant solutions to Nahm's equations, with $H$ being a connected, nontrivial Lie subgroup of $\SO (3)$, thus either $H \simeq \SO (2)$, or $H = \SO (3)$. Motivated by spatial geometry, we make the following definitions:

\smallskip

\begin{definition}\label{definition:symmetries}
	When $H \simeq \SO (2)$, the solution is called \emph{axially symmetric}. Similarly, when $H = \SO (3)$, the solution is called \emph{spherically symmetric}.
\end{definition}

\smallskip

\begin{remark}
	The conserved quantities in \cref{eq:CC1,eq:CC2,eq:CC3,eq:CC4,eq:CC5} transform according to the 5-dimensional irreducible representation of $\SO (3)$ under the action of $\SO (3)$, and are invariant under the action of $\SU (n)$; see for instance \cite{Dancer-NahmHyperKahler}*{Equation~(21)}.

	More precisely, let $\cT$ be a solution to Nahm's equation and let us redefine the corresponding conserved quantities as
	\begin{equation}
		\forall i, j \in \{ 1, 2, 3 \}: \quad C_{ij} \left( \cT \right) \coloneqq \tr \left( T_i T_j \right) - \frac{1}{3} \sum\limits_{k = 1}^2 \tr \left( T_k^2 \right). \label{eq:C_ij}
	\end{equation}
	Note that \cref{eq:C_ij} defines a 3-by-3, real, symmetric, and traceless matrix, and its 5 independent parameters can be chosen to be the conserved quantities in \cref{eq:CC1,eq:CC2,eq:CC3,eq:CC4,eq:CC5}. Let us denote this matrix by $\cC \left( \cT \right)$. Then for all $A \in \SO (3)$ and $U \in \SU (n)$ we have
	\begin{equation}
		\cC \left( {}_U^A \cT \right) = A \cC \left( \cT \right) A^{- 1}.
	\end{equation}
	In particular, if $\cT$ is spherically symmetric, then all conserved quantities in \cref{eq:CC1,eq:CC2,eq:CC3,eq:CC4,eq:CC5} must vanish by Schur's Lemma. In fact, any $\tr(T(\zeta)^k)=0$ vanish for the same reason and so the spectral curve of any spherically symmetric monopole is just $\eta^n=0$.
\end{remark}

\smallskip

In either of the above cases, $H$-invariance implies the existence of a function, $U : H \rightarrow \SU (n)$, such that
\begin{equation}
	\forall A \in H : \quad {}^A \cT = {}_{U (A)^{- 1}} \cT.
\end{equation}
However, as in \Cref{rem:UA_properties}, this function need not be homomorphic or smooth, and moreover, it need not be unique. Indeed, if the stabilizer subgroup of $\cT$
\begin{equation}
	S_\cT \coloneqq \{ U \in \SU (n) \ | \ {}_U \cT = \cT \},
\end{equation}
is nontrivial, then one can choose a function $H \rightarrow S_\cT$ and ``twist''. Hence, $U$ may not be unique, and moreover, even if $U$ was homomorphic or smooth, the twisted version may not have these properties.

\smallskip

The case in which $U$ can be chosen to have some regularity, at least around the identity, is easier to handle. For this reason, our main theorems require that $U$ is continuously differentiable at the identity.

\bigskip

\section{Axially symmetric solutions}
\label{sec:axial}

In this section, we consider axially symmetric solutions, that is, when $H$ is isomorphic to $\SO (2)$ in \Cref{definition:symmetries}.

\smallskip

All subgroups of $\SO (3)$ that are isomorphic to $\SO (2)$ are maximal tori of $\SO (3)$ and thus are conjugate to each other. Moreover, they can be viewed as rotations around a given, oriented axis (that is, an oriented line through the origin). Hence, without loss of generality, it is enough to study one of them. Let
\begin{equation}
	A: \rl \rightarrow \SO (3); \quad \theta \mapsto A (\theta) \coloneqq \begin{pmatrix} \cos (\theta) & - \sin (\theta) & 0 \\ \sin (\theta) & \cos (\theta) & 0 \\ 0 & 0 & 1 \end{pmatrix}. \label{eq:form_of_A}
\end{equation}
Then our choice of such a subgroup is
\begin{equation}
	H \coloneqq \left\{ A (\theta) \ \middle| \ \theta \in [0, 2 \pi) \right\}. \label{eq:form_of_H}
\end{equation}
This subgroup is the group of rotations around the third axis. Moreover, \cref{eq:form_of_A} provides a global parametrization of $H$.

\smallskip

Our first theorem gives an infinitesimal version of axial symmetry.

\begin{theorem}
	\label{theorem:Axial_Symmetry}
	If $\cT = \left( T_1, T_2, T_3 \right)$ is axially symmetric around the third axis and the corresponding $U$ function in \eqref{eq:A-equivariance} can be chosen to be continuously differentiable at the identity of $\SO (3)$, then there exists a $Y \in \su (n)$ such that
	\begin{subequations}
	\begin{align}
		T_1	&= [T_2, Y], \label{eq:Linearized_Axial_Symmetry_1} \\
		T_2 &= [Y, T_1], \label{eq:Linearized_Axial_Symmetry_2} \\
		0		&= [Y, T_3]. \label{eq:Linearized_Axial_Symmetry_3}
	\end{align}
	\end{subequations}
	Conversely, if \cref{eq:Linearized_Axial_Symmetry_1,eq:Linearized_Axial_Symmetry_2,eq:Linearized_Axial_Symmetry_3} are satisfied, then $\cT$ is axially symmetric around the third axis.

	More generally, $\cT$ is axially symmetric around some axis if there exists $B \in \SO (3)$ such that $^B \cT$ satisfies \cref{eq:Linearized_Axial_Symmetry_1,eq:Linearized_Axial_Symmetry_2,eq:Linearized_Axial_Symmetry_3}. 
\end{theorem}

\begin{definition}
	We call $Y$ the \emph{generator of the axially symmetry for $\cT$}.
\end{definition}

\begin{proof}
	Assume that $\cT = \left( T_1, T_2, T_3 \right)$ is axially symmetric around the third axis and the corresponding $U$ function in \eqref{eq:A-equivariance} can be chosen so that it is continuously differentiable at the identity of $\SO (3)$, and let $Y \coloneqq \tfrac{\rd}{\rd \theta} \left( U_{A (\theta)} \right)\big|_{\theta = 0}$. Then \cref{eq:Linearized_Axial_Symmetry_1,eq:Linearized_Axial_Symmetry_2,eq:Linearized_Axial_Symmetry_3} are just the linearizations of
	\begin{equation}
		{}_{U_{A (\theta)}}^{A (\theta)} \cT = \cT,
	\end{equation}
	at $\theta = 0$.

	On the other hand, if \cref{eq:Linearized_Axial_Symmetry_1,eq:Linearized_Axial_Symmetry_2,eq:Linearized_Axial_Symmetry_3} hold for some $Y \in \su (n)$, then, for all $\theta \in \rl$, let $A (\theta)$ be defined via \cref{eq:form_of_A} and let
	\begin{equation}
		U (\theta) \coloneqq \exp \left( \theta Y \right) \in \SU (n). \label{eq:form_of_U}
	\end{equation}
	Now \cref{eq:Linearized_Axial_Symmetry_1,eq:Linearized_Axial_Symmetry_2,eq:Linearized_Axial_Symmetry_3} are equivalent to
	\begin{equation}
		^{\dot{A} (0)} \cT +[Y, \cT] = 0.
	\end{equation}
	Now simple computation, using \cref{eq:form_of_A,eq:form_of_U}, shows that for all $\theta$,
	\begin{equation}
		\frac{\rd}{\rd \theta} \left( {}_{U (\theta)}^{A (\theta)} \cT \right) = {}^{\dot{A} (\theta)}_{U (\theta)} \cT + {}^{A (\theta)}_{\dot{U} (\theta)} \cT = {}^{\dot{A} (0)} \left( {}_{U (\theta)}^{A (\theta)} \cT \right) + \left[ Y, {}_{U (\theta)}^{A (\theta)} \cT \right] = {}_{U (\theta)}^{A (\theta)} \left( {}^{\dot{A} (0)}\cT +[Y, \cT] \right) = 0.
	\end{equation}
	Hence
	\begin{equation}
		\forall \theta \in \rl : \quad \cT = {}_{U (0)}^{A (0)} \cT = {}_{U (\theta)}^{A (\theta)} \cT,
	\end{equation}
	which is equivalent to $\cT$ being axially symmetric around the third axis.

	The last claim follows from the discussion in the beginning of this section.
\end{proof}

\smallskip

The moral of \Cref{theorem:Axial_Symmetry} is that, up to rotation and gauge, axially symmetric solutions to Nahm's equations are labeled by elements $Y \in \su (n)$. Note that after a gauge transformation of $\cT$, the corresponding $Y$ changes by the adjoint action of $\SU (n)$. Thus, we only need to consider the adjoint orbits in $\su (n)$. It is easy to find canonical representatives in every orbit: In every orbit there is a unique $Y$ of the form $Y = i \cdot \diag (\alpha_1, \alpha_2, \ldots, \alpha_n)$ with $\alpha_k \in \rl$, $\alpha_k \geqslant \alpha_{k + 1}$. Of course, $\sum_{i = 1}^n \alpha_i = 0$ has to also hold, as elements of $\su (n)$ are traceless. However, there are adjoint orbits that cannot carry a nontrivial Nahm datum satisfying \cref{eq:Linearized_Axial_Symmetry_1,eq:Linearized_Axial_Symmetry_2,eq:Linearized_Axial_Symmetry_3}, as shown in the next lemma.

\begin{lemma}
	\label{lemma:negative_1}
	If $\cT$ is a nonconstant, axially symmetric solution to Nahm's \cref{eq:Nahm1,eq:Nahm2,eq:Nahm3}, and $Y \in \su (n)$ is the generator of the axially symmetry for $\cT$, then $-1 \in \Spec \left( \ad_Y^2 \right)$.

	Conversely, if $- 1 \in \Spec \left( \ad_Y^2 \right)$, then there is a nonconstant, axially symmetric solution to Nahm's equations, whose generator of the axial symmetry is $Y$.
\end{lemma}

\begin{remark}
	Note that $- 1$ is in the spectrum of $\ad_Y^2$, if and only if $\alpha_j = \alpha_i + 1$, for some $1 \leqslant i < j \leqslant n$.
\end{remark}

\begin{proof}[Proof of \Cref{lemma:negative_1}]
	\Cref{eq:Linearized_Axial_Symmetry_1,eq:Linearized_Axial_Symmetry_2} imply that
	\begin{equation}
		\ad_Y^2 \left( T_1 \right) = - T_1, \quad \mbox{and} \quad \ad_Y^2 \left( T_2 \right) = - T_2.
	\end{equation}
	This proves the first claim.

	For the converse, we first show that \cref{eq:Linearized_Axial_Symmetry_1,eq:Linearized_Axial_Symmetry_2,eq:Linearized_Axial_Symmetry_3} can be satisfied at a point, call $t_0$. Note that $\ad_Y$ always has a nontrivial kernel, because $0 \neq Y \in \ker (\ad_Y)$. Pick any nonzero element $T \in \ker (\ad_Y)$ and let $T_3 (t_0) = T$. As $- 1 \in \Spec \left( \ad_Y^2 \right)$, we can choose $T_1 (t_0)$ to be a $(- 1)$-eigenvector and let $T_2 (t_0) = \ad_Y \left( T_1 (t_0) \right)$. The Picard--Lindel\"of Theorem guarantees the existence of a (local) solution with these initial values. As $\tfrac{\rd}{\rd t} \cT$ is nonzero at $t = t_0$, $\cT$ is necessarily nonconstant. This concludes the proof.
\end{proof}

\smallskip

Using \cref{eq:Linearized_Axial_Symmetry_2}, $T_2$ can be eliminated from Nahm's equations. Since $[[Y, T_1], T_3] = [Y, [T_1, T_3]]$, the reduced system of ordinary differential equations, which we call the \emph{axially symmetric Nahm's equations}, has the form:
\begin{subequations}
\begin{align}
	\dot{T}_1	&= [Y, [T_1, T_3]], \label{eq:Axial_Symmetric_Nahm1} \\
	\dot{T}_3	&= [T_1, [Y, T_1]]. \label{eq:Axial_Symmetric_Nahm2}
\end{align}
\end{subequations}

\medskip

\subsection{Axially symmetric solutions to the $\mathbf{SU (3)}$ Nahm's equations}

We now illustrate how \Cref{theorem:Axial_Symmetry,lemma:negative_1} can be used to find axially symmetric solutions to Nahm's equations through the $n = 3$ case. Now $Y \in \su (3)$ and thus, after a change of basis, it can always be brought to the form
\begin{equation}
	Y_{\alpha, \beta} \coloneqq \begin{pmatrix}
									\alpha i & 0 & 0 \\
				 					0 & \beta i & 0 \\
									0 & 0 & - (\alpha + \beta) i \\
								\end{pmatrix}, \mbox{ with } \alpha, \beta \in \rl, \mbox{ and } \alpha \geqslant \beta \geqslant - (\alpha + \beta).
\end{equation}
By \Cref{lemma:negative_1}, to get nonconstant solutions, we need to have
\begin{equation}
	\beta = \alpha - 1, \quad \mbox{or} \quad \beta = \tfrac{1}{2} (1 - \alpha), \quad \mbox{or} \quad \beta = 1 - 2 \alpha.
\end{equation}
Let us write a generic element of $\su (3)$, say $T$, as
\begin{equation}
	T = \begin{pmatrix} a i & z_1 & z_2\\ -\overline{z_1} & b i & z_3\\ -\overline{z_2} & - \overline{z_3} & -(a + b) i \\ \end{pmatrix}, \mbox{ with } a, b \in \rl, \mbox{ and } z_1, z_2, z_3 \in \cx.
\end{equation}
Then we get
\begin{equation}
	\ad_{Y_{\alpha, \beta}}^2 (T) =
	\begin{pmatrix}
		0 & -(\alpha - \beta)^2 z_1 & -(2 \alpha + \beta)^2 z_2\\
		(\alpha - \beta)^2 \overline{z_1} & 0 & - (\alpha + 2 \beta)^2 z_3\\
		(2 \alpha + \beta)^2 \overline{z_2} & (\alpha + 2 \beta)^2 \overline{z_3} & 0\\
	\end{pmatrix}. \label{eq:ad2YT}
\end{equation}
Thus, \cref{eq:ad2YT} imply that, when no two of the diagonal elements of $Y_{\alpha, \beta}$ are equal, then the kernel is spanned by elements of the form $i \cdot \diag (a, b, - (a + b))$, and thus this is the most general form of $T_3$, in this case. When two of the diagonal elements are equal, then the kernel is 4-dimensional and isomorphic (as a Lie algebra) to $\rl \oplus \so (3)$.

\smallskip

\Cref{table:dim_-1_eigenspaces} below summarizes the cases in which the $(- 1)$-eigenspaces are nontrivial (with the requirement that $\alpha \geqslant \beta \geqslant - (\alpha + \beta)$).
\begin{table}[h!]
\begin{center}
	\begin{tabular}{|c|c|c|c|}
		\hline
		case \# & $(\alpha, \beta)$ & example(s) & $\dim_{\rl} \left( \ker \left( \ad_{Y_{\alpha, \beta}}^2 + \Id_{\su (n)} \right) \right)$ \\
		\hline
		1. & $\left( \alpha, \alpha - 1 \right), \: \alpha \in \left( \tfrac{2}{3}, 1 \right) \cup \left(1, \infty \right)$ & $\left( \tfrac{5}{6}, - \tfrac{1}{6} \right), \left( 2, 1 \right)$ & 2 \\
		2. & $\left( \alpha, 1 - 2 \alpha \right), \: \tfrac{1}{3} < \alpha < \tfrac{2}{3}$ & $\left( \tfrac{1}{2}, 0 \right)$ & 2 \\
		3. & $\left( \alpha, \tfrac{1}{2} (1 - \alpha) \right), \: \alpha \in \left( \tfrac{1}{3}, 1 \right) \cup \left( 1, \infty \right)$ & $\left( \tfrac{2}{3}, \tfrac{1}{6} \right), \left( 3, - 1 \right)$ & 2 \\
		4. & $\left( 1, 0 \right)$ & $\left( 1, 0 \right)$ & 4 \\
		5. & $\left( \tfrac{2}{3}, - \tfrac{1}{3} \right)$ & $\left( \tfrac{2}{3}, - \tfrac{1}{3} \right)$ & 4 \\
		6. & $\left( \tfrac{1}{3}, \tfrac{1}{3} \right)$ & $\left( \tfrac{1}{3}, \tfrac{1}{3} \right)$ & 4 \\
		\hline
	\end{tabular}
\end{center}
\caption{Dimensions of the $(- 1)$-eigenspaces.}
\label{table:dim_-1_eigenspaces}
\end{table}

In the cases 1., 2., and 3., the Ansatz has $2 + 2 = 4$ real parameters, and the axially symmetric Nahm's \cref{eq:Axial_Symmetric_Nahm1,eq:Axial_Symmetric_Nahm2} reduce to a system of four ordinary differential equation on four real functions. In the cases 4., 5., and 6., the Ans\"atze have $4 + 4 = 8$ parameters, and the axially symmetric Nahm's \cref{eq:Axial_Symmetric_Nahm1,eq:Axial_Symmetric_Nahm2} reduce to a system of eight ordinary differential equation on eight real functions.

\smallskip

We end this section by computing the solutions explicitly in a particular case.

\subsubsection*{Example: The $(\alpha, \beta) = \left( \tfrac{1}{2}, 0 \right)$ case}

In this example we show how our technique recovers the results of \cite{Dancer-NahmHyperKahler}*{Proposition~3.10}.

\smallskip

Let $Y \coloneqq Y_{\frac{1}{2},0}$, that is
\begin{equation}
Y =	\begin{pmatrix}
		\tfrac{i}{2} & 0 & 0 \\
		0 & 0 & 0 \\
		0 & 0 & - \tfrac{i}{2} \\
	\end{pmatrix}.
\end{equation}
Both the kernel and the $(- 1)$-eigenspace of $\ad_Y^2$ are 2-dimensional. More concretely, we can write our Ansatz as
\begin{equation}
	T_1 = \begin{pmatrix} 0 & 0 & z \\ 0 & 0 & 0 \\ - \overline{z} & 0 & 0 \end{pmatrix} \quad \& \quad T_3 = \begin{pmatrix} a i & 0 & 0 \\ 0 & - (a + b) i & 0 \\ 0 & 0 & b i \end{pmatrix},
\end{equation}
where $z$ is a complex function and $a$ and $b$ are real functions. Using the residual gauge symmetry, we can assume, without any loss of generality, that $z$ is, in fact, real. The axially symmetric Nahm's \cref{eq:Axial_Symmetric_Nahm1,eq:Axial_Symmetric_Nahm2} then become
\begin{subequations}
\begin{align}
	\dot{z}	&= (a - b) z, \label{eq:case2_eq1} \\
	\dot{a}	&= 2 z^2, \label{eq:case2_eq2} \\
	\dot{b}	&= - 2 z^2. \label{eq:case2_eq3}
\end{align}
\end{subequations}
Note that if one knows $a$ and $b$, then $z$ can be computed via \cref{eq:case2_eq1}. The conserved quantities in \cref{eq:CC1,eq:CC2,eq:CC3} are automatically zero. The other two, given by \cref{eq:CC4,eq:CC5}, are related and satisfy
\begin{equation}
	C_4 = C_5 = a^2 + b^2 + (a + b)^2 - 2 z^2.
\end{equation}
Furthermore, \cref{eq:case2_eq2,eq:case2_eq3} imply that the quantities
\begin{equation}
	k_1 \coloneqq a + b, \quad \mbox{and} \quad k_2 \coloneqq a^2 - k_1 a - z^2,
\end{equation}
are also conserved. Using \cref{eq:case2_eq2} we get
\begin{equation}
	\dot{a} = 2 \left( a^2 - k_1 a - k_2 \right) = 2 \left( a - \frac{1}{2} k_1 \right)^2 - 2 \left( k_2 + \frac{k_1^2}{4} \right). \label{eq:dot_a_eq}
\end{equation}
Let
\begin{equation}
	A \coloneqq 2 a - k_1 = a - b, \quad \mbox{and} \quad K \coloneqq - 4 k_2 - k_1^2.
\end{equation}
Then \cref{eq:dot_a_eq} is equivalent to
\begin{equation}
	\dot{A} = A^2 + K. \label{eq:Simplified_Eq_1}
\end{equation}
Let $c \in \mathbb{R}$ be the constant of integration, and then the solutions of \cref{eq:Simplified_Eq_1} are
\begin{equation}
	A (t) = \left\{ 
		\begin{array}{cc}
		- \sqrt{K} \cot \left( \sqrt{K} (t - c) \right) & K > 0, \\
		- (t - c)^{- 1} & K = 0, \\
		- \sqrt{- K} \coth \left( \sqrt{- K} (t - c) \right) & K < 0.
	\end{array} \right. \label{eq:axial_example}
\end{equation}
From this the functions $a, b$, and $z$, and thus the corresponding axially symmetric solution of Nahm's equations can easily be reconstructed.

\begin{remark}
	In each case of \cref{eq:axial_example}, the solutions develop singularities. When $K > 0$, the singularities are at the points
	\begin{equation}
		\left\{ \ c + \tfrac{\pi}{K} k \ \middle| \ k \in \Z \ \right\}.
	\end{equation}
	When $K \leqslant 0$, the only singularity is at $t = c$.

	According to \Cref{lemma:poles_and_representations} these singularities induce a representation of $\so (3)$. In all of the three cases of \cref{eq:axial_example} this representation is the (unique, up to isomorphism) irreducible, 3-dimensional representation.
\end{remark}

\bigskip

\section{Spherically symmetric solutions}
\label{sec:spherical}

In this section we consider spherically symmetric solutions to Nahm's equations. That is, solutions $\cT$, such that for all $A \in \SO (3)$ there exists $U_A \in \SU (n)$ such that $_{U_A}^A \cT = \cT$.

\smallskip

The next theorem is analogous to \Cref{theorem:Axial_Symmetry} as it gives an infinitesimal version of spherical symmetry.

\begin{theorem}
	\label{theorem:Spherical_Symmetry}
	If $\cT = \left( T_1, T_2, T_3 \right)$ is spherically symmetric and the corresponding $U$ function in \eqref{eq:A-equivariance} can be chosen to be continuously differentiable at the identity of $\SO (3)$, then there exists a triple, $\left( Y_1, Y_2, Y_3 \right) \in \su (n)$, such that for all $i, j \in \{ 1, 2, 3 \}$
	\begin{equation}
		[Y_i, T_j] = \sum\limits_{k = 1}^3 \epsilon_{ijk} T_k, \label{eq:Linearized_Spherical_Symmetry}
	\end{equation}
	Conversely, if \cref{eq:Linearized_Spherical_Symmetry} is satisfied, then $\cT$ is spherically symmetric.
\end{theorem}

\begin{definition}
	We call elements of the triple $\left( Y_1, Y_2, Y_3 \right)$ the \emph{generators of the spherically symmetry for $\cT$}.
\end{definition}

\begin{proof}
	Assume first that $\cT$ is spherically symmetric and the corresponding $U$ function in \eqref{eq:A-equivariance} can be chosen so that it is continuously differentiable at the identity of $\SO (3)$. In this case, $\cT$ is, in particular, axially symmetric around all of the three coordinate axes and the corresponding $U$ functions in \eqref{eq:A-equivariance} can be chosen so that they are continuously differentiable at the identity of $\SO (3)$, thus the first part of \Cref{theorem:Axial_Symmetry} can be applied. For each $i \in \{ 1, 2, 3 \}$, let $X_i \in \so (3)$ be the infinitesimal generator of the (positively oriented) rotation around the $i^\mathrm{th}$ axis, that is $(X_i)_{jk} = - \epsilon_{ijk}$, and let $Y_i$ be the corresponding generator of axial symmetry. Then the nine equations in \eqref{eq:Linearized_Spherical_Symmetry} are exactly the three triples of equations that one gets from these three different axial symmetries.

	\smallskip

	Now assume that \cref{eq:Linearized_Spherical_Symmetry} holds and let $A \in \SO (3)$. Due to the surjectivity of the exponential map $\exp : \so (3) \rightarrow \SO (3)$, there exists $X = a_1 X_1 + a_2 X_2 + a_3 X_3 \in \so (3)$, such that $A = \exp (X)$. Let $A (s) \coloneqq \exp (s X) \in \SO (3)$, $Y \coloneqq a_1 Y_1 + a_2 Y_2 + a_3 Y_3 \in \su (n)$, and $U (s) \coloneqq \exp (s Y)$. Clearly, $A (1) = A$, $A (0) = \Id_3$, and $U (0) = \Id_n$. Thus $_{U (0)}^{A (0)} \cT = \cT$. Next we show that $\cA (s) \coloneqq {}_{U (s)}^{A (s)} \cT - \cT$ is independent of $s$. Using
	\begin{equation}
		\frac{\rd}{\rd s} A (s) = X A (s) = A (s) X, \quad \mbox{and} \quad \frac{\rd}{\rd s} U (s) = Y U (s) = U (s) Y,
	\end{equation}
	we get
	\begin{equation}
		\frac{\rd}{\rd s} \cA (s) = {}_{U (s)}^{A (s)} \left( {}^X \cT + [Y, \cT] \right).
	\end{equation}
	Now the vanishing of the $s$-independent quantity in the parentheses is equivalent to \cref{eq:Linearized_Spherical_Symmetry}. Thus, $\cA$ is constant and since $\cA (0) = 0$, we have $\cA (1) = 0$, which is equivalent to $\cT$ being $A$-equivariant. Since $A \in \SO(3)$ was arbitrary, this concludes the proof.
\end{proof}

\medskip

\begin{remark}
	Note that \cref{eq:Linearized_Spherical_Symmetry} implies that for all $i, j, l \in \{ 1, 2, 3 \}$
	\begin{align}
		[[Y_i, Y_j], T_l]	&= [Y_i, [Y_j, T_l]] - [Y_j, [Y_i, T_l]] \\
			&= \sum\limits_{l = 1}^3 \left( \epsilon_{jlk} [Y_i, T_k] - \epsilon_{ilk} [Y_j, T_k] \right) \\
			&= \sum\limits_{k, m = 1}^3 \left( \epsilon_{jlk} \epsilon_{ikm} - \epsilon_{ilk} \epsilon_{jkm} \right) T_m \\
			&= \sum\limits_{m = 1}^3 \left( \delta_{ik} \delta_{jm} - \delta_{im} \delta_{jk} \right) T_m \\
			&= \sum\limits_{l, m = 1}^3 \epsilon_{ijk} \epsilon_{klm} T_m \\
			&= \sum\limits_{l = 1}^3 \epsilon_{ijk} [Y_k, T_l],
	\end{align}
	or, equivalently
	\begin{equation}
		[[Y_i, Y_j], \cT] = \sum\limits_{k = 1}^3 \epsilon_{ijk} [Y_k, \cT]. \label{eq:Y_commutators}
	\end{equation}
\end{remark}

\medskip

\subsection{Structure theorem for the spherically symmetric Ansatz}

In this section we prove a structure theorem for spherically symmetric solutions of Nahm's equations, under certain hypotheses. This structure theorem classifies spherically symmetric An\"atze through representation theoretic means. In order to set up the stage for this, let us begin with a remark.

\begin{remark}\label{remark:representations}
	If a triple $\left( Y_1, Y_2, Y_3 \right) \in \su (n)^{\oplus 3}$ satisfies the $\so (3)$ commutator relations
	\begin{equation}
		\forall i \in \{ 1, 2, 3 \}: \quad [Y_i, Y_j] = \sum\limits_{k = 1}^3 \epsilon_{ijk} Y_k, \label{eq:so3commutators}
	\end{equation}
	then \cref{eq:Y_commutators} holds independent of $\cT$. A compact way to rephrase \cref{eq:so3commutators} can be given as follows: consider the linear map from $\so (3)$ to $\su (n)$ that sends $X_i$ to $Y_i$. By \cref{eq:so3commutators}, this map gives a representation of $\so (3)$. Let us call it $(\cx^n, \rho)$. Let $\left( V_k, \rho_k \right)$ be the irreducible, $k$-dimensional complex representation of $\so (3)$ and let
	\begin{equation}
		\left( \hat{V}, \hat{\rho} \right) \coloneqq (\cx^n, \rho) \otimes \left( (\cx^n)^*, \rho^* \right) \otimes \left( V_3, \rho_3 \right).
	\end{equation}
	The values of $\cT$ can be viewed as elements of $\hat{V}$ (more precisely elements from $\su (n) \otimes V_3 \subsetneq \hat{V}$) and the action of any $X \in \so (3)$ on $\cT$ is given by
	\begin{equation}
		\hat{\rho} (X) (\cT) \coloneqq {}^X \cT + [\rho (X), \cT]. \label{eq:hat_rho}
	\end{equation}
	If $\cT$ is spherically symmetric with generators $\left( Y_1, Y_2, Y_3 \right)$, then \cref{eq:Linearized_Spherical_Symmetry} is equivalent to
	\begin{equation}
		\forall X \in \so (3): \quad \hat{\rho} (X) (\cT) = 0, \label{eq:spherical_symmetry_with_hat_rho}
	\end{equation}
	or, in other words, the values of $\cT$ lie in the trivial component of $\left( V, \hat{\rho} \right)$. Since only representation of $\so (3)$ that is both trivial and irreducible is the 1-dimensional one, \Cref{theorem:Spherical_Symmetry} implies the following: When \cref{eq:so3commutators} holds, $\cT$ is spherically symmetric exactly if $\cT$ takes values in the direct sum of the 1-dimensional irreducible components of $\left( \hat{V}, \hat{\rho} \right)$.
\end{remark}

\smallskip

\Cref{theorem:Spherical_Symmetry,remark:representations} tell us that certain spherically symmetric solutions to Nahm's equations are labeled by representations of $\so (3)$.

Let us recall the Clebsch--Gordan Theorem: For all positive integers, $m \geqslant n$, we have the following decomposition of representations
\begin{equation}
	\left( V_m, \rho_m \right) \otimes \left( V_n, \rho_n \right) \simeq \smashoperator{\bigoplus\limits_{k = 1}^n} \ \left( V_{m + n + 1 - 2 k}, \rho_{m + n + 1 - 2 k} \right). \label{eq:Clebsch--Gordan_Decomposition}
\end{equation}
The above \cref{eq:Clebsch--Gordan_Decomposition} implies that the representation $\left( V_m, \rho_m \right) \otimes \left( V_n, \rho_n \right) \otimes \left( V_3, \rho_3 \right)$ has a single 1-dimensional irreducible summand when $m = n \geqslant 2$, or $m = n + 2$, and none otherwise. When $m = n + 2$, let us call $B^n$ the unit length generator of the unique 1-dimensional representation in $\left( V_{n + 2}, \rho_{n + 2} \right) \otimes \left( V_n, \rho_n \right) \otimes \left( V_3, \rho_3 \right)$, which is well-defined, up to a $\rU (1)$ factor. Note, that $B^n$ can be viewed as an $\so (3)$-invariant triple of maps $(B^n_1,B^n_2,B^n_3)$ from the $n$-dimensional irreducible representation to the $(n + 2)$-dimensional one. With this in mind, let us state and prove our structure theorem.

\begin{theorem}[Structure theorem for spherically symmetric solutions to Nahm's equations]
\label{theorem:Str_Thm_for_Spherical_Symmetry}
	Let $\cT$ be a spherically symmetric solution to Nahm's \cref{eq:Nahm1,eq:Nahm2,eq:Nahm3} such that the generators induce the representation $(\cx^n, \rho)$ (as in \Cref{remark:representations}), and write the decomposition of $(\cx^n, \rho)$ into irreducible summands as
	\begin{equation}
		(\cx^n, \rho) \simeq \bigoplus\limits_{a = 1}^k \left( V_{n_a}, \rho_{n_a} \right). \label{eq:rho_decomposition}
	\end{equation}
	We can assume, without any loss of generality, that for all $a \in \{ 1, \ldots, k - 1 \}$, we have $n_a \geqslant n_{a + 1}$. For all $i \in \{ 1, 2, 3 \}$ and $a \in \{ 1, \ldots, k \}$ let
	\begin{align}
		Y_{i, a}	&\coloneqq \rho_{n_a} (X_i), \\
		Y_i			&\coloneqq \rho (X_i) = \diag \left( Y_{i, 1}, Y_{i, 2}, \ldots, Y_{i, k} \right).
	\end{align}
	Fix $t_0$ in the domain of $\cT$ and for all $i \in \{ 1, 2, 3 \}$, write $T_i (t_0)$ in a block matrix form at any point in the domain of $\cT$, according to \cref{eq:rho_decomposition}\emph{:}
	\begin{equation}
		T_i (t_0) = \begin{pmatrix}
			\left( T_i \right)_{11}	& \left( T_i \right)_{12}	& \ldots	& \left( T_i \right)_{1k} \\
			\left( T_i \right)_{21} & \left( T_i \right)_{22}	&			& \vdots \\
			\vdots	 				&							& \ddots	& \\
			\left( T_i \right)_{k1} & \ldots					&			& \left( T_i \right)_{kk}
		\end{pmatrix}, \label{eq:Spherically_Symmetric_Ansatz}
	\end{equation}
	with $\left( T_i \right)_{ab} = - \left( T_i \right)_{ba}^* \in V_{n_a} \otimes V_{n_b}^*$ and $\left( T_i \right)_{aa} \in \su \left( V_{n_a} \right)$. Then, up to a $\rho$-invariant gauge and for all $a, b \in \{ 1, 2, \ldots, k \}$, we have that
	\begin{enumerate}
		\item There exist $c_a \in \rl$, such that for all $i \in \{ 1, 2, 3 \}$, we have $\left( T_i \right)_{aa} = c_a Y_{i, a}$.
		\item If $a \neq b$ and $|n_a - n_b| \neq 2$, then $\left( T_i \right)_{ab} = 0$.
		\item If $n_a = n_b + 2$, then there exists $c_{ab} \in \cx$, such that for all $i \in \{ 1, 2, 3 \}$, we have $\left( T_i \right)_{ab} = c_{ab} B_i^{n_a}$.\footnote{Recall that $B^n$ is a unit length generator of the unique 1-dimensional representation in $\left( V_{n + 2}, \rho_{n + 2} \right) \otimes \left( V_n, \rho_n \right) \otimes \left( V_3, \rho_3 \right)$.}
	\end{enumerate}
\end{theorem}

\smallskip

\begin{remark}
\label{remark:spherical_reducibility}
	An important question in any gauge theory is the reducibility of solutions. If $\cT$ is a spherically symmetric Nahm datum, then its decomposition into irreducible components is easy to read from \cref{eq:Spherically_Symmetric_Ansatz}: $\cI \coloneqq \{ 1, \ldots, k \}$, and define a relation on $\cI$ via $a \sim_0 b$ exactly when there exists $i \in \{ 1, 2, 3 \}$, such that $\left( T_i \right)_{ab} \neq 0$. Note that $\sim_0$ is reflexive and symmetric (but need not be transitive). Thus it generates an equivalence relation on $I$ which we denote by $\sim$. The set of irreducible components of $\cT$ is then labeled by $\cJ \coloneqq \cI / \sim$ and the component corresponding to $\{ a_1, a_2, \ldots, a_l \} \in \cJ$ is given by
	\begin{equation}
		\left( \cT_{\{ a_1, a_2, \ldots, a_l \}} \right)_i (t_0) \coloneqq \begin{pmatrix}
			\left( T_i \right)_{a_1 a_1}	& \left( T_i \right)_{a_1 a_2}	& \ldots	& \left( T_i \right)_{a_1 a_l} \\
			\left( T_i \right)_{a_2 a_1}	& \left( T_i \right)_{a_2 a_2}	&			& \vdots \\
			\vdots	 						&								& \ddots	& \\
			\left( T_i \right)_{a_l a_1}	& \ldots						&			& \left( T_i \right)_{a_l a_l}
		\end{pmatrix}.
	\end{equation}
	In particular, we have the following:
	\begin{enumerate}
		
		\item The third bullet point in \Cref{theorem:Str_Thm_for_Spherical_Symmetry} implies that if the set $\{ n_1, n_2, \ldots, n_k \}$ contains both even and odd numbers, then $\cT$ is reducible.

		\item If $k = 2$, $n_1 = n_2 + 2$, and for some $i \in \{ 1, 2, 3 \}$, $\left( T_i \right)_{12} \neq 0$, then $\cT$ is irreducible.

	\end{enumerate}
\end{remark}

\smallskip

\begin{remark}
	The $k = 1$ case in \Cref{theorem:Str_Thm_for_Spherical_Symmetry}, that is when $(\cx^n, \rho)$ is irreducible, was studied, albeit for low ranks only, by Dancer in \cite{Dancer-SU3monopoles}. This solution is also equivalent to the one found in the $K = 0$ case of \cref{eq:axial_example}.
\end{remark}

\smallskip

\begin{proof}[Proof of \Cref{theorem:Str_Thm_for_Spherical_Symmetry}]
	Let us begin with the $k = 1$ case, that is when $(\cx^n, \rho)$ is irreducible. By \Cref{remark:representations} the value of $\cT$ at any point is an element in the trivial component of $\left( \hat{V}, \hat{\rho} \right)$. Then the Clebsch--Gordan decomposition, \cref{eq:Clebsch--Gordan_Decomposition}, tells us that there is a unique 1-dimensional trivial summand in $\left( \hat{V}, \hat{\rho} \right)$, thus if we find one Ansatz, then it is the most general one. Let $\rho (X_i) \coloneqq Y_i$. Note that if for all $i \in \{ 1, 2, 3 \}$ we have $T_i = f Y_i$ at any point of the domain of $\cT$, then we have that for all, and it is a spherically symmetric Ansatz. This completes the proof in the $k = 1$ case.

	Let us turn to the general case. Regard $\cT$ as an element of $\hat{V}$ via evaluating it at $t_0$. Recall that $\so (3)$ acts on $V$ via (the extension of)
	\begin{equation}
		\hat{\rho} : \so (3) \otimes \hat{V} \rightarrow \hat{V}; \quad X \otimes \cR \mapsto {}^X \cR + [\rho (X), \cR].
	\end{equation}
	By \Cref{theorem:Spherical_Symmetry,remark:representations}, $\cT$ is spherically symmetric if \cref{eq:spherical_symmetry_with_hat_rho} holds, or, in other words, if $\cT$ takes values in the 1-dimensional irreducible summands of $\left( \hat{V}, \hat{\rho} \right)$. We can write the decomposition of $\left( \hat{V}, \hat{\rho} \right)$ into irreducible components as
	\begin{equation}
		\left( \hat{V}, \hat{\rho} \right) = \bigoplus\limits_{a, b = 1}^k \left( V_{n_a} \otimes V_{n_b} \otimes V_3, \rho_{n_a} \otimes \rho_{n_b} \otimes \rho_3 \right).
	\end{equation}
	Since $\cT$ takes values in $\hat{V}$, we get \cref{eq:Spherically_Symmetric_Ansatz}. Moreover the values of $\cT$ are in fact $\su (n) \otimes V_3$, we get $\left( T_i \right)_{ab} = - \left( T_i \right)_{ba}^* \in V_{n_a} \otimes V_{n_b}^*$ and $\left( T_i \right)_{aa} \in \mathfrak{u} \left( V_{n_a} \right)$. Using \cref{eq:Linearized_Spherical_Symmetry} we also see that $\left( T_i \right)_{aa}$ is traceless, thus $\left( T_i \right)_{aa} \in \su \left( V_{n_a} \right)$.

	Let $\left( \hat{V}_0, \hat{\rho}_0 \right)$ be the direct sum of 1-dimensional irreducible components in $\left( \hat{V}, \hat{\rho} \right)$. By the proof of the $k = 1$ case, each summand of the diagonal terms has a unique copy of $\left( V_0, \rho_0 \right)$, proving the first bullet point. Furthermore, each off-diagonal term also has unique copy, if and only if $|n_b - n_c| = 2$, proving the second bullet point. Finally, one can easily verify that when $n_a = n_b$ for some $a \neq b$, then there exists $U \in \SU (n)$, such that $[U, Y] = 0$, $U$ only acts nontrivially on the $\left( V_{n_a} \oplus V_{n_b}, \rho_{n_a} \oplus \rho_{n_b} \right)$ summand, and the $ab$-component of ${}_U \cT$ is zero. This completes the proof.
\end{proof}

\smallskip

\begin{remark}
	A weaker version of \Cref{theorem:Str_Thm_for_Spherical_Symmetry} was stated, without rigorous proof, in \cite{BCGPS83}.
\end{remark}

\smallskip

Next we analyze the off-diagonal terms in \cref{eq:Spherically_Symmetric_Ansatz}. As before, let $(X_1, X_2, X_3)$ be a standard basis of $\so (3)$, that is for all $i, j, k \in \{ 1, 2, 3 \}$, we have $(X_i)_{jk} = - \epsilon_{ijk}$.

\smallskip

\begin{theorem}
	\label{theorem:matrix_identities}
	Fix $n \in \N_+$, and let $\left( V_{n + 2}, \rho_{n + 2} \right)$ and $\left( V_n, \rho_n \right)$ be as above. For all $i \in \{ 1, 2, 3 \}$, let $Y_i^+ \coloneqq \rho_{n + 2} (X_i)$, $Y_i^- \coloneqq \rho_n (X_i)$, and $Y_i \coloneqq \diag \left( Y_i^+, Y_i^- \right)$. Let $\left( B_1^n, B_2^n, B_3^n \right)$ be the (unique up to a $\rU (1)$ factor) triple in the (unique) 1-dimensional irreducible component of
	\begin{equation}
		\left( \hat{V}, \hat{\rho} \right) \coloneqq \left( V_{n + 2} \otimes V_n \otimes V_3, \rho_{n + 2} \otimes \rho_n \otimes \rho_3 \right).
	\end{equation}
	Then, after a potential rescaling\footnote{That is, replacing $B^n$ with $\lambda B^n$, for some $\lambda \in \cx - \{ 0 \}$.}, we have the following identities (for all $j, k \in \{ 1, 2, 3 \}$, where it applies):
	\begin{subequations}
	\begin{align}
		Y_j^+ B_k^n - B_k^n Y_j^-									&= \sum\limits_{i = 1}^3 \epsilon_{ijk} B_i^n, \label{eq:B_i_def} \\
		\sum\limits_{i = 1}^3 B_i^n \left( B_i^n \right)^* 			&= \Id_{n + 2}, \label{eq:B_norm_1} \\
		\sum\limits_{i = 1}^3 \left( B_i^n \right)^* B_i^n 			&= \frac{n + 2}{n} \Id_n, \label{eq:B_norm_2} \\
		B_j^n \left( B_k^n \right)^* - B_k^n \left( B_j^n \right)^*	&= - \frac{2}{n + 1} \sum\limits_{i = 1}^3 \epsilon_{ijk} Y_i^+, \label{eq:YB_relations_1} \\
		\left( B_j^n \right)^* B_k^n - \left( B_k^n \right)^* B_j^n	&= \frac{2 (n + 2)}{n (n + 1)} \sum\limits_{i = 1}^3 \epsilon_{ijk} Y_i^-, \label{eq:YB_relations_2} \\
		Y_j^+ B_k^n - Y_k^+ B_j^n									&= \frac{n + 3}{2} \sum\limits_{i = 1}^3 \epsilon_{ijk} B_i^n, \label{eq:YB_relations_3} \\
		B_j^n Y_k^- - B_k^n Y_j^-									&= - \frac{n - 1}{2} \sum\limits_{i = 1}^3 \epsilon_{ijk} B_i^n. \label{eq:YB_relations_4}
	\end{align}
	\end{subequations}
\end{theorem}

\begin{proof}
	\Cref{eq:B_i_def} is the defining equation for $\left( B_1^n, B_2^n, B_3^n \right)$, hence needs no proof.

	Using \cref{eq:B_i_def}, for all $j \in \{ 1, 2, 3 \}$ we get that
	\begin{align}
		\left[ Y_j^+, \sum\limits_{i = 1}^3 B_i^n \left( B_i^n \right)^* \right]	&= 0, \\
		\left[ Y_j^-, \sum\limits_{i = 1}^3 \left( B_i^n \right)^* B_i^n \right]	&= 0,
	\end{align}
	thus, by Schur's Lemma, we have there are real numbers, say $a_n$ and $b_n$, such that
	\begin{align}
		\sum\limits_{i = 1}^3 B_i^n \left( B_i^n \right)^*	&= a_n \Id_{n + 2}, \\
		\sum\limits_{i = 1}^3 \left( B_i^n \right)^* B_i^n	&= b_n \Id_n.
	\end{align}
	By taking traces we can conclude that both $a_n$ and $b_n$ are positive. Since $\left( B_1^n, B_2^n, B_3^n \right)$ can be rescaled by a nonzero constant, we can achieve $a_n = 1$. Again by the properties of the trace, we get $b_n = \tfrac{n + 2}{n}$. This proves \cref{eq:B_norm_1,eq:B_norm_2}.

	For all $i \in \{ 1, 2, 3 \}$ let
	\begin{align}
		\widetilde{Y}_i^+	&\coloneqq \sum\limits_{j, k = 1}^3 \epsilon_{ijk} B_j^n \left( B_k^n \right)^*, \label{eq:tildeY_i^+_def} \\
		\widetilde{Y}_i^-	&\coloneqq \sum\limits_{j, k = 1}^3 \epsilon_{ijk} \left( B_j^n \right)^* B_k^n, \label{eq:tildeY_i^-_def} \\
		\widetilde{B}_i^+	&\coloneqq \sum\limits_{j, k = 1}^3 \epsilon_{ijk} Y_j^+ B_k^n, \label{eq:tildeB_i^+_def} \\
		\widetilde{B}_i^-	&\coloneqq \sum\limits_{j, k = 1}^3 \epsilon_{ijk} B_j^n Y_k^-. \label{eq:tildeB_i^-_def}
	\end{align}
	Using \cref{eq:B_i_def}, we get that for all $i, j \in \{ 1, 2, 3 \}$
	\begin{align}
		\left[ Y_i^+, \widetilde{Y}_j^+ \right]								&= \sum\limits_{k = 1}^3 \epsilon_{ijk} \widetilde{Y}_k^+, \\
		\left[ Y_i^-, \widetilde{Y}_j^- \right]								&= \sum\limits_{k = 1}^3 \epsilon_{ijk} \widetilde{Y}_k^-, \\
		Y_i^+ \widetilde{B}_j^\pm - \widetilde{B}_j^\pm Y_i^-	&= \sum\limits_{k = 1}^3 \epsilon_{ijk} \widetilde{B}_k^\pm.
	\end{align}
	Thus, by the proof of the $k = 1$ case in \Cref{theorem:Str_Thm_for_Spherical_Symmetry} and the uniqueness (up to scale) of $\left( B_1^n, B_2^n, B_3^n \right)$, we get that that there are real numbers $\alpha_n^\pm$ and $\beta_n^\pm$, such that for all $i \in \{ 1, 2, 3 \}$
	\begin{align}
		\widetilde{Y}_i^\pm	&= \alpha_n^\pm Y_i^\pm, \\
		\widetilde{B}_i^\pm	&= \beta_n^\pm B_i^n.
	\end{align}
	Next we four, independent, linear equations on $\left( \alpha_n^+, \alpha_n^-, \beta_n^+, \beta_n^- \right) \in \rl^4$ to be able to (uniquely) solve for them. Let
	\begin{align}
		C_n^+ \Id_{n + 2}	&\coloneqq \sum\limits_{i = 1}^3 \left( Y_i^+ \right)^2 = - \frac{(n + 1)(n + 3)}{4} \Id_{n + 2}, \label{eq:Casimir+} \\
		C_n^- \Id_n			&\coloneqq \sum\limits_{i = 1}^3 \left( Y_i^- \right)^2 = - \frac{n^2 - 1}{4} \Id_n, \label{eq:Casimir-}
	\end{align}
	be the Casimir operators corresponding to the two irreducible representations. Using \cref{eq:tildeY_i^+_def,eq:tildeY_i^-_def,eq:Casimir+,eq:Casimir-}, we get
	\begin{equation}
		\alpha_n^+ C_n^+ \Id_{n + 2} = \sum\limits_{i = 1}^3 Y_i^+ (\alpha_n^+ Y_i^+) = \sum\limits_{i = 1}^3 Y_i^+ \widetilde{Y}_i^+ = \sum\limits_{i, j, k = 1}^3 \epsilon_{ijk} Y_i^+ B_j^n \left( B_k^n \right)^* = \sum\limits_{k = 1}^3 \widetilde{B}_k^+ \left( B_k^n \right)^* = \beta_n^+ \Id_{n + 2},
	\end{equation}
	which gives
	\begin{equation}
		\alpha_n^+ C_n^+ - \beta_n^+ = 0. \label{eq:system_1}
	\end{equation}
	Similarly
	\begin{equation}
		\alpha_n^- C_n^- \Id_n = \sum\limits_{k = 1}^3 (\alpha_n^- Y_k^-) Y_k^- = \sum\limits_{k = 1}^3 \widetilde{Y}_k^- Y_k^- = \sum\limits_{i, j, k = 1}^3 \epsilon_{ijk} \left( B_i^n \right)^* B_j^n Y_k^- = \sum\limits_{i = 1}^3 \left( B_i^n \right)^* \widetilde{B}_i^- = \beta_n^- \frac{n + 2}{n} \Id_n,
	\end{equation}
	so, we get
	\begin{equation}
		\alpha_n^- C_n^- - \frac{n + 2}{n} \beta_n^- = 0. \label{eq:system_2}
	\end{equation}
	Next, adding up \cref{eq:B_i_def,eq:tildeB_i^+_def,eq:tildeB_i^-_def}, and using \cref{eq:B_i_def} yields
	\begin{equation}
		\beta_n^+ + \beta_n^- = 2. \label{eq:system_3}
	\end{equation}
	Let us define
	\begin{equation}
		B \coloneqq \sum\limits_{i = 1}^3 Y_i^+ B_i^n = \sum\limits_{i = 1}^3 B_i^n Y_i^-. \label{eq:B_def}
	\end{equation}
	Using \cref{eq:B_i_def}, we get $Y_i^+ B - B Y_i^- = 0$ for all $i \in \{ 1, 2, 3 \}$, and thus, by Schur's Lemma, $B = 0$. Using this and \cref{eq:tildeB_i^+_def}, for all $i \in \{ 1, 2, 3 \}$ we have
	\begin{align}
		\left( \beta_n^+ \right)^2 B_i^n	&= \sum\limits_{j, k = 1}^3 \epsilon_{ijk} Y_j^+ \left( \beta_n^+ B_k^n \right) \\
											&= \sum\limits_{j, k, l, m = 1}^3 \epsilon_{ijk} \epsilon_{klm} Y_j^+ Y_l^+ B_m^n \\
											&= \sum\limits_{j, l, m = 1}^3 \left( \delta_{il} \delta_{jm} - \delta_{im} \delta_{jl} \right) Y_j^+ Y_l^+ B_m^n \\
											&= \sum\limits_{j = 1}^3 \left( Y_j^+ Y_i^+ B_j^n - Y_j^+ Y_j^+ B_i^n \right) \\
											&= \left( \sum\limits_{j, k = 1}^3 \epsilon_{jik} Y_k^+ B_j^n \right) + Y_i^+ \left( \sum\limits_{j = 1}^3 Y_j^+ B_j^n \right) - C_n^+ B_i^n \\
											&= \beta_n^+ B_i + 0 - C_n^+,
	\end{align}
	thus we get
	\begin{equation}
		\left( \beta_n^+ \right)^2 = \beta_n^+ - C_n^+. \label{eq:beta_quad_+}
	\end{equation}
	An analogous computation gives
	\begin{equation}
		\left( \beta_n^- \right)^2 = \beta_n^- - C_n^-. \label{eq:beta_quad_-}
	\end{equation}
	Subtracting \cref{eq:beta_quad_-} from \cref{eq:beta_quad_+}, and using \cref{eq:system_3} gives
	\begin{equation}
		\left( \beta_n^+ \right)^2 - \left( \beta_n^- \right)^2 = \left( \beta_n^+ + \beta_n^- \right) \left( \beta_n^- - \beta_n^- \right) = 2 \left( \beta_n^- - \beta_n^- \right) = \left( \beta_n^- - \beta_n^- \right) - C_n^+ + C_n^-,
	\end{equation}
	and thus
	\begin{equation}
		\beta_n^+ - \beta_n^- = C_n^- - C_n^+ = \frac{n^2 + 4 n + 3 - (n^2 - 1)}{4} = 2 (n + 2). \label{eq:system_4}
	\end{equation}
	The \cref{eq:system_1,eq:system_2,eq:system_3,eq:system_4} form a set of four independent linear equations for four unknowns, so the solution is unique:
	\begin{equation}
		\alpha_n^+ = - \frac{2}{n + 1}, \quad \alpha_n^- = \frac{2 (n + 2)}{n (n + 1)}, \quad \beta_n^+ = \frac{n + 3}{2}, \quad \beta_n^- = - \frac{n - 1}{2}.
	\end{equation}
	This completes the proof of \cref{eq:YB_relations_1,eq:YB_relations_2,eq:YB_relations_3,eq:YB_relations_4}.
\end{proof}

\medskip

\subsection{Ans\"atze and solutions}
In the next theorem, using \Cref{theorem:matrix_identities} we investigate a generalization of the irreducible solutions defined in the second bullet point of \Cref{remark:spherical_reducibility}.

\begin{theorem}
	\label{theorem:long_chains}
	Let $n, k \geqslant 1$, and suppose that $\left( V, \rho \right)$ decomposes as
	\begin{equation}\label{eq:Vrho_decomp}
	\begin{aligned}
		\left( V, \rho \right) \simeq \smashoperator{\bigoplus\limits_{a = 0}^k} \left( V_{n + 2 k - 2 a}, \rho_{n + 2 k - 2 (a - 1)} \right) =\left( V_{n + 2 k}, \rho_{n + 2 k} \right) \oplus \ldots \oplus \left( V_{n + 2}, \rho_{n + 2} \right) \oplus \left( V_n, \rho_n \right).
	\end{aligned}
	\end{equation}
	For all $m \geqslant 1$ and $i \in \{ 1, 2, 3 \}$, let $Y_i^m \coloneqq \rho_m (X_i)$ and $\left( B_1^m, B_2^m, B_3^m \right)$ as in \Cref{theorem:matrix_identities}.

	Then the spherically symmetric Ansatz given in \cref{eq:Spherically_Symmetric_Ansatz} takes the following form: there exists real, analytic functions $f_0, f_1, \ldots, f_k, g_0, g_1, \ldots, g_{k - 1}$ on the domain of $\cT$ such that for all $i \in \{ 1, 2, 3 \}$ we have
	\begin{equation}
		T_i = \begin{pmatrix}
			f_k Y_i^{n + 2 k} 									& g_{k - 1} B_i^{n + 2 (k - 1)}						& 0 							& \ldots						& 0			\\
			- g_{k - 1} \left( B_i^{n + 2 (k - 1)} \right)^*	& f_{k - 1} Y_i^{n + 2 (k - 1)}						& g_{k - 2} B_i^{n + 2 (k - 2)}	& 0								& \vdots	\\
			0													& - g_{k - 2} \left( B_i^{n + 2 (k - 2)} \right)^*	& \ddots						& 								& 0			\\
			\vdots												& 0													&								& \ddots						& g_0 B_i^n	\\
			0													& \ldots											& 0								& - g_0 \left( B_i^n \right)^*	& f_0 Y_i^n
		\end{pmatrix}. \label{eq:spherical_symmetric_long_chain}
	\end{equation}
	By convention, let $g_{- 1} = g_k \equiv 0$.

	When $n = 1$, then $\left( Y_1^1, Y_2^1, Y_3^1 \right) = (0, 0, 0)$, and thus $f_0$ can be omitted. The rest of the functions above satisfy the following equations:
	\begin{subequations}
	\begin{align}
		\forall a \in \{ 1, 2, \ldots, k \}: \quad \dot{f}_a		&= f_a^2 + \frac{1}{a} g_{a - 1}^2 - \frac{2 a + 3}{(2 a + 1) (a + 1)} g_a^2, \label{eq:dot_f_a^1} \\
		\forall a \in \{ 0, 1, \ldots, k - 1 \}: \quad \dot{g}_a	&= \left( (a + 2) f_{a + 1} - a f_a \right) g_a. \label{eq:dot_g_a^1}
	\end{align}
	\end{subequations}

	When $n \geqslant 2$, then the functions above satisfy the following equations:
	\begin{subequations}
	\begin{align}
		\forall a \in \{ 0, 1, \ldots, k \}: \quad \dot{f}_a		&= f_a^2 + \frac{2}{n + 2 a - 1} g_{a - 1}^2 - \frac{2 (n + 2 a + 2)}{(n + 2 a) (n + 2 a + 1)} g_a^2, \label{eq:dot_f_a^n} \\
		\forall a \in \{ 0, 1, \ldots, k - 1 \}: \quad \dot{g}_a	&= \left( \frac{n + 2 a + 3}{2} f_{a + 1} - \frac{n + 2 a - 1}{2} f_a \right) g_a. \label{eq:dot_g_a^n}
	\end{align}
	\end{subequations}
\end{theorem}

Note that \Cref{eq:dot_f_a^1,eq:dot_g_a^1} are simply \Cref{eq:dot_f_a^n,eq:dot_g_a^n} with $n=1$, with the exception that the differential equation for $f_0$ is omitted.

\begin{proof}
	When $\left( V, \rho \right)$ decomposes as in \cref{eq:Vrho_decomp}, \cref{eq:Spherically_Symmetric_Ansatz} in \Cref{theorem:Str_Thm_for_Spherical_Symmetry} reduces to \cref{eq:spherical_symmetric_long_chain}. The \cref{eq:dot_f_a^1,eq:dot_g_a^1,eq:dot_f_a^n,eq:dot_g_a^n} follow using \cref{eq:B_i_def,eq:YB_relations_1,eq:YB_relations_2,eq:YB_relations_3,eq:YB_relations_4}.
\end{proof}

\smallskip

\begin{remark}
	\label{remark:affine_action}
	Let $\cT = \left( T_1, T_2, T_3 \right)$ be a solution to Nahm's \cref{eq:Nahm1,eq:Nahm2,eq:Nahm3} on an open, connected interval $I \subseteq \rl$, and $f$ be a nonconstant affine function on $\rl$, that is for all $t \in \rl$, $f (t) = a t + b$, for some $a, b \in \rl$ with $a \neq 0$. We define the \emph{pull-back} of $\cT$ via $f$ as
	\begin{equation}
		t \mapsto f^* (\cT) (t) \coloneqq \left( a T_1 (f (t)), a T_2 (f(t)), a T_3 (f(t)) \right),
	\end{equation}
	Then $f^* (\cT)$ is also a solution to Nahm's equations on the interval $f^{- 1} (I)$.

	In particular, if $\cT$ has domain $I = (\alpha, \beta)$, then choosing $f (t) = \tfrac{\beta - \alpha}{2} t + \tfrac{\alpha + \beta}{2}$ yields that $f^* (\cT)$ has domain $(- 1, 1)$. Furthermore, the residues of $f^* (\cT)$ at $\alpha$ (resp. at $\beta$) are exactly $a$ times the residues of $\cT$ at $\alpha$ (resp. at $\beta$).

	Thus, we can assume that the domain of $\cT$ is $(- 1, 1)$, albeit at the price of losing two degrees of freedom in the general solution.
\end{remark}

\smallskip

In the following three theorems we use \Cref{theorem:long_chains} in the cases when $n = 1$ and $k \in \{ 1, 2 \}$ and when $n \geqslant 2$ and $k = 1$, to find solutions to Nahm's equations.

\begin{theorem}
	\label{theorem:3+1Monopoles}
	Under the hypotheses and notation of \Cref{theorem:long_chains} let $n = 1$ and $k = 1$. Thus $\left( V, \rho \right) \simeq \left( V_3, \rho_3 \right) \oplus \left( V_1, \rho_1 \right)$. Let $\cT$ be a spherically symmetric solution to Nahm's \cref{eq:Nahm1,eq:Nahm2,eq:Nahm3}, with representation induced by $\left( V, \rho \right)$. Let the domain of $\cT$ be a connected, open interval, $I$.

	In this case, choose $\rho_3$ to be the identity and $\left( B_1^1, B_2^1, B_3^1 \right)$ to be the standard (orthonormal and oriented) basis of $\rl^3$. Then, in some gauge, there are real, analytic functions, $f = f_2$ and $g = g_1$, on $I$, such that for all $i \in \{ 1, 2, 3 \}$
	\begin{equation}
		T_i = \begin{pmatrix}
			f X_i^3						& g B_i^1 \\
			- g \left( B_i^1 \right)^*	& 0
		\end{pmatrix}, \label{eq:spherical_symmetric_Ti_3+1}
	\end{equation}
	and $f$ and $g$ satisfy the following equations:
	\begin{subequations}
	\begin{align}
		\dot{f}	&= f^2 + g^2, \label{eq:3+1_f_eq} \\
		\dot{g}	&= 2 f g. \label{eq:3+1_g_eq}
	\end{align}
	\end{subequations}
	Any maximally extended, irreducible solution to \cref{eq:3+1_f_eq,eq:3+1_g_eq} develops poles in both direction, thus the domain of such a solution is necessarily a bounded, open interval, and the only maximally extended, irreducible solution\footnote{There is another solution with $g \equiv 0$ which is reducible to the direct sum of the rank 3 copy of the irreducible solutions found in the $K = 0$ case of \cref{eq:axial_example} and a trivial solution.} to \cref{eq:3+1_f_eq,eq:3+1_g_eq} with the boundary conditions that the residues are exactly at $t = \pm 1$ is
	\begin{subequations}
	\begin{align}
		f (t)	&= \frac{t}{1 - t^2}, \label{eq:3+1_f} \\
		g (t)	&= \frac{1}{1 - t^2}, \label{eq:3+1_g}
	\end{align}
	\end{subequations}
	with residues given by
	\begin{subequations}
	\begin{align}
		\Res \left( f, \pm 1 \right)	&= - \frac{1}{2}, \label{eq:3+1_res_f} \\
		\Res \left( g, \pm 1 \right)	&= - \frac{1}{2}. \label{eq:3+1_res_g}
	\end{align}
	\end{subequations}
	The spherically symmetric Nahm datum given by \cref{eq:spherical_symmetric_Ti_3+1,eq:3+1_f,eq:3+1_g}, induces representations that isomorphic to $(V_2,\rho_2)^{\oplus 2}$ at both poles.
\end{theorem}

\begin{proof}
	First of all, \cref{eq:spherical_symmetric_Ti_3+1} is a special case of \cref{eq:spherical_symmetric_long_chain} and \cref{eq:3+1_f_eq,eq:3+1_g_eq} are special cases of \cref{eq:dot_f_a^1,eq:dot_g_a^1}, with $n = 1$, $k = 1$, $f = f_1$, and $g = g_0$. Simple commutations give \cref{eq:3+1_f,eq:3+1_g}\footnote{For example, \cref{eq:dot_f_a^1,eq:dot_g_a^1} can be decoupled via the substitution $F_\pm = f \pm g$.} and the \cref{eq:3+1_res_f,eq:3+1_res_g}. The claim about the irreducibility follows from \cref{eq:res_g} and the second bullet point in \Cref{remark:spherical_reducibility}.

	The decomposition of the poles can be verified using the Casimir operators
	\begin{equation}
		\hat{C}^\pm \coloneqq - \sum\limits_{i = 1}^3 \Res \left( T_i, \pm 1 \right)^2. \label{eq:casimir_at_pm1_3+1}
	\end{equation}
	Substitute $f$ and $g$ from \cref{eq:3+1_res_f,eq:3+1_res_g} into \cref{eq:casimir_at_pm1_3+1} to get
	\begin{equation}
		\hat{C}^\pm = \frac{2^2 - 1}{4} \Id_{4}.
	\end{equation}
	Both Casimir operators are proportional to the identity, thus the representations are either irreducible, or reducible to isomorphic summands. By inspecting the dimensions and the factors of proportionality we find that the latter is true. In fact, we get that the induced representations are isomorphic to $\left( V_2, \rho_2 \right)^{\oplus 2}$ at both poles, which completes the proof.
\end{proof}

\smallskip

To demonstrate the same idea in higher rank cases, we present, without proof (as the claim is easy to check), the case of $n = 1$ and $k = 2$ in the following theorem.

\begin{theorem}
	\label{theorem:5+3+1Monopoles}
	Under the hypotheses and notation of \Cref{theorem:long_chains}), let $\left( V, \rho \right) \simeq \left( V_5, \rho_5 \right) \oplus \left( V_3, \rho_3 \right) \oplus \left( V_1, \rho_1 \right)$. Let $\cT$ be a spherically symmetric solution to Nahm's \cref{eq:Nahm1,eq:Nahm2,eq:Nahm3}, with representation induced by $\left( V, \rho \right)$.

	Then, in some gauge, there are real, analytic functions, $f_2, f_3, g_1$, and $g_2$, on the domain of $\cT$, such that for all $i \in \{ 1, 2, 3 \}$
	\begin{equation}
		T_i = \begin{pmatrix}
			f_2 Y_i^5						& g_1 B_i^3						& 0			\\
			- g_1 \left( B_i^3 \right)^*	& f_1 Y_i^3						& g_0 B_i^1	\\
			0								& - g_0 \left( B_i^1 \right)^*	& 0
		\end{pmatrix},
	\end{equation}
	and $f_1, f_2, g_0$, and $g_1$ satisfy the following equations
	\begin{equation}
		\dot{f}_1 = f_1^2 + g_0^2 - \frac{5}{6} g_1^2, \quad \dot{f}_2 = f_2^2 + \frac{1}{2} g_1^2, \quad \dot{g}_0 = 2 f_1 g_0, \quad \dot{g}_1 = \left( 3 f_2 - f_1 \right) g_1.
	\end{equation}
	The above equations have the following maximally extended, irreducible solution\footnote{However this solution need not be (and probably is not) unique, in any sense.} on $(- 1, 1)$:
	\begin{equation}
		f_2 (t) = f_3 (t) = -\frac{t}{t^2-1}, \quad g_1 (t) = \frac{\sqrt{8}}{\sqrt{3}(t^2-1)}, \quad g_2 (t) = \frac{\sqrt{2}}{t^2-1},
	\end{equation}
	with residues given by
	\begin{equation}
		\Res \left( f_2, \pm 1 \right) = \Res \left( f_3, \pm 1 \right) = -\frac{1}{2}, \quad \Res \left( g_1, \pm 1 \right) = \pm\sqrt{\frac{2}{3}}, \quad \Res \left( g_2, \pm 1 \right) = \pm\frac{1}{\sqrt{2}}.
	\end{equation}
	The above Nahm datum induces representations that are isomorphic to $]\left(V_3, \rho_3 \right)^{\oplus 3}$ at both poles.
\end{theorem}

\smallskip

Finally, investigate the case of $n \geqslant 2$ and $k = 1$.

\begin{theorem}
	\label{theorem:(n+2)+nMonopoles}
	Under the hypotheses and notation of \Cref{theorem:long_chains} let $n \geqslant 2$ and $k = 1$, thus $\left( V, \rho \right) \simeq \left( V_{n + 2}, \rho_{n + 2} \right) \oplus \left( V_n, \rho_n \right)$, and let $\cT$ be a spherically symmetric solution to Nahm's \cref{eq:Nahm1,eq:Nahm2,eq:Nahm3}, with representation induced by $\left( V, \rho \right)$. Let the domain of $\cT$ be a connected, open interval, $I$. Finally, let $\left( Y_1^\pm, Y_2^\pm, Y_3^\pm \right)$ and $\left( B_1, B_2, B_3 \right)$ as in \Cref{theorem:Str_Thm_for_Spherical_Symmetry}.

	Then, in some gauge, there are real, analytic functions, $f^\pm$ and $g$, on $I$, such that for all $i \in \{ 1, 2, 3 \}$
	\begin{equation}
		T_i = \begin{pmatrix}
			f^+ Y_i^+ & g B_i \\
			- g B_i^* & f^- Y_i^-
		\end{pmatrix}, \label{eq:spherical_symmetric_Ti_(n+2)+n}
	\end{equation}
	and they satisfy the following system of ordinary differential equations
	\begin{subequations}
	\begin{align}
		\dot{f}^+	&= \left( f^+ \right)^2 + \frac{2}{n + 1} g^2, \label{eq:general_(n+2)+n_1} \\
		\dot{g}		&= \left( \frac{n + 3}{2} f^+ - \frac{n - 1}{2} f^- \right) g, \label{eq:general_(n+2)+n_2} \\
		\dot{f}^-	&= \left( f^- \right)^2 - \frac{2 (n + 2)}{n (n + 1)} g^2. \label{eq:general_(n+2)+n_3}
	\end{align}
	\end{subequations}
	The above equations have the following maximally extended, irreducible solution\footnote{However this solution need not be (and probably is not) unique, in any sense.} on $(- 1, 1)$:
	\begin{subequations}
	\begin{align}
		f^+ (t)	&= - \frac{(n + 1) t + (n - 1)}{(n + 1)(t^2 - 1)}, \label{eq:f+} \\
		f^- (t)	&= - \frac{(n + 1) t + (n + 3)}{(n + 1)(t^2 - 1)}, \label{eq:f-} \\
		g (t)	&= \left( \frac{2 n}{n + 1} \right)^{\tfrac{1}{2}} \frac{1}{t^2 - 1}, \label{eq:g}
	\end{align}
	\end{subequations}
	with residues given by
	\begin{subequations}
	\begin{align}
		\Res \left( f^+, \pm 1 \right)	&= \frac{(2 - n) \mp (n + 2)}{2 (n + 1)}, \label{eq:res_f+} \\
		\Res \left(f^-, \pm 1 \right)	&= \frac{1}{2} \mp \frac{n + 3}{2 (n + 1)}, \label{eq:res_f-} \\
		\Res \left(g, \pm 1 \right)		&= \pm \sqrt{\frac{n}{2 (n + 1)}}. \label{eq:res_g}
	\end{align}
	\end{subequations}
	The spherically symmetric Nahm datum given by \cref{eq:spherical_symmetric_Ti_(n+2)+n,eq:f+,eq:f-,eq:g}, induces representations that isomorphic to $\left( V_{n + 1}, \rho_{n+1} \right)^{\oplus 2}$ at $t = 1$ and to $\left( V_2, \rho_2 \right)^{\oplus (n + 1)}$ at $t = - 1$.
\end{theorem}

\begin{proof}
	First of all, \cref{eq:spherical_symmetric_Ti_(n+2)+n} is a special case of \cref{eq:spherical_symmetric_long_chain} and \cref{eq:spherical_symmetric_Ti_(n+2)+n,eq:general_(n+2)+n_1,eq:general_(n+2)+n_2,eq:general_(n+2)+n_3} are special cases of \cref{eq:dot_f_a^1,eq:dot_g_a^1}, with $n \geqslant 2$, $k = 1$, $f_0 = f^-$, $f_1 = f^+$, and $g \coloneqq g_0$. Checking \cref{eq:f+,eq:f-,eq:g,eq:res_f+,eq:res_f-,eq:res_g} then is straightforward. The claim about the irreducibility follows from \cref{eq:res_g} and the second bullet point in \Cref{remark:spherical_reducibility}.

	In order to find the decomposition of the induced representations of $\mathfrak{so}(3)$ at $t = \pm 1$, we compute the corresponding Casimir operators:
	\begin{equation}
		\hat{C}^\pm \coloneqq - \sum\limits_{i = 1}^3 \Res \left( T_i, \pm 1 \right)^2. \label{eq:casimir_at_pm1}
	\end{equation}
	For each $i \in \{ 1, 2, 3 \}$ and $t \in (- 1, 1)$, we have
	\begin{equation}
		T_i (t)^2 = \begin{bmatrix}
			f^+ (t)^2 \left( Y_i^+ \right)^2 - g (t)^2 B_i B_i^*					& f^+ (t) g (t) Y_i^+ B_i+ f^- (t) g (t) B_i Y_i^- \\
			 f^+ (t) g (t) B_i^* Y_i^+ + f^- (t) g (t) Y_i^- B_i^*	& - g (t)^2 B_i^* B_i + f^- (t)^2 \left( Y_i^- \right)^2
		\end{bmatrix}. \label{eq:T_i_squared}
	\end{equation}
	Using \cref{eq:B_norm_1,eq:B_norm_2,eq:Casimir+,eq:Casimir-,eq:B_def,eq:T_i_squared} we get
	\begin{equation}
		- \sum\limits_{i = 1}^3 T_i (t)^2 = \begin{bmatrix}
			\left( \frac{(n + 1)(n + 3)}{4} f^+ (t)^2 + g (t)^2 \right) \Id_{n + 2}	& 0 \\
			0	& \left( \frac{n^2 - 1}{4} f^- (t)^2 + \frac{n + 2}{n} g (t)^2 \right) \Id_n
		\end{bmatrix}. \label{eq:T_i_squared_simplified}
	\end{equation}
	Now substituting $f^+, f^-$, and $g$ from \cref{eq:res_f+,eq:res_f-,eq:res_g} into \cref{eq:T_i_squared_simplified} and taking residues, we get
	\begin{equation}
		\hat{C}^+ = \frac{(n + 1)^2 - 1}{4} \Id_{2 n + 2}, \quad \mbox{and} \quad \hat{C}^- = \frac{2^2 - 1}{4} \Id_{2 n + 2}.
	\end{equation}
	Both Casimir operators are proportional to the identity, thus the representations are either irreducible, or reducible to isomorphic summands. By inspecting the dimensions and the factors of proportionality we find that the latter is true. In fact, we get that the induced representations are isomorphic to $\left( V_{n + 1}, \rho_{n+1} \right)^{\oplus 2}$ at $t = 1$ and to $\left( V_2, \rho_2 \right)^{\oplus (n + 1)}$ at $t = - 1$, which completes the proof.
\end{proof}

\bigskip

\section{The Nahm transform}
\label{sec:Nahm_transform}

In this section we use the solutions to Nahm's equations that found in the previous sections, to construct solutions to the BPS monopole equation, using the \emph{Nahm transform}, which was was introduced by Nahm in \cite{N82}, and it is a version of the Atiyah--Drinfeld--Hitchin--Manin construction; cf. \cite{ADHM78}.

\medskip

\subsection{Brief summary of the technique}

We outline the Nahm transform of Nahm data defined on a single interval. Let $\cT = \left( T_1, T_2, T_3 \right)$ be a rank $N$ Nahm data that is analytic on the nonempty, open interval $(a, b)$, and satisfies certain boundary conditions at $a$ and $b$. The precise form of these boundary conditions can be found in \cite{HurMurMpoleConstClassGrps} in the maximal symmetry breaking case, and in \cite{Charbonneau-Nagy-NahmTransform} in the general case. We only consider the case when $\cT$ has a pole at both boundary points, and the residue forms (potentially reducible) $N$-dimensional representation of $\so (3)$. Let $e_1$, $e_2$, and $e_3$ represent multiplication by the standard unit quaternions on $\mathbb{C}^2 \cong \mathbb{H}$, and let $L_{1,0}^2 \left( (a, b) ; \mathbb{C}^N \otimes \mathbb{H} \right)$ be the Hilbert space of $L_1^2$ functions with Dirichlet boundary values. For each $\mathbf{x} \in \mathbb{R}^3$, we define a first-order differential operator $\Lambda_{\mathbf{x}}^\cT$, called the \emph{Nahm--Dirac operator}, to be
\begin{align}
	\Lambda_{\mathbf{x}}^\cT : L_{1,0}^2 \left( (a, b); \mathbb{C}^N \otimes \mathbb{H} \right) &\rightarrow L^2 \left( (a, b) ; \mathbb{C}^N \otimes \mathbb{H} \right); \\
	f &\mapsto i \frac{\rd f}{\rd t} + \sum\limits_{a = 1}^3 \left( \left( i T_a - x_a \Id_{\mathbb{C}^N} \right) \otimes e_a \right) f.
\end{align}
Note that $\Lambda_{\mathbf{x}}^\cT$ is a Fredholm operator. Moreover, $\left\{ \ \Lambda_{\mathbf{x}}^\cT \ \middle| \ \mathbf{x} \in \rl^3 \ \right\}$ is a norm-continuous family of Fredholm operators; in particular, the index is independent of $\mathbf{x} \in \rl^3$.

\smallskip

Nahm's main observation in \cite{N82} was that if $\cT$ is a solution to Nahm's equations, then for all $\mathbf{x} \in \rl^3$, the kernel of $\Lambda_{\mathbf{x}}^\cT$ is trivial. Thus the index of $\Lambda_{\mathbf{x}}^\cT$ is nonpositive, and the cokernel-bundle, defined via $E_{\mathbf{x}}^\cT \coloneqq \coker \left( \Lambda_{\mathbf{x}}^\cT \right)$, is a Hermitian vector bundle over $\rl^3$, with a distinguished connection induced by the product connection on the trivial Hilbert-bundle $\rl^3 \times L^2 \left( (a, b); \mathbb{C}^N \otimes \mathbb{H} \right)$. Let us denote this connection by $\nabla^\cT$, let $\Pi_{E_{\mathbf{x}}^\cT}$ be the orthogonal projection from $L_{1,0}^2 \left( (a, b); \mathbb{C}^N \otimes \mathbb{H} \right)$ to $E_{\mathbf{x}}^\cT$, and let
\begin{align}
	m : L_{1,0}^2 \left( (a, b) ; \mathbb{C}^N \otimes \mathbb{H} \right) &\rightarrow L_{1,0}^2 \left( (a, b) ; \mathbb{C}^N \otimes \mathbb{H} \right); \\
	f &\mapsto \left( t \mapsto i t f (t) \right).
\end{align}
Finally, let us define $\Phi^\cT \in \End \left( E^\cT \right)$ as
\begin{equation}
	\forall \mathbf{x} \in \rl^3: \quad \Phi_{\mathbf{x}}^\cT \coloneqq \Pi_{\mathbf{x}}^\cT \circ m - \tfrac{\tr_{E^\cT} \left( \Pi_{\mathbf{x}}^\cT \circ m \right)}{\rk \left( E^\cT \right)} \Id_{E^\cT}.
\end{equation}
Note that $\Phi^\cT \in \su \left( E^\cT \right)$. The main results of \cites{N82,HurMurMpoleConstClassGrps,Charbonneau-Nagy-NahmTransform} is that the pair $(\nabla^\cT, \Phi^\cT)$ is a (finite energy) BPS monopole on $E^\cT$. Moreover, the asymptotic behavior of this monopole is governed by the poles of $\cT$ at $a$ and $b$.

\medskip

\subsection{Construction of spherically symmetric BPS monopoles}
\label{subsec:construction}

In this section we use \Cref{theorem:3+1Monopoles,theorem:(n+2)+nMonopoles} to generate spherically symmetric monopoles. We employ Maple (see \cite{Maplecode} for details) to carry out the Nahm transform. More precisely, we solve the following ordinary differential equation for all $r \in \rl_+$:
\begin{equation}
	\left( \Lambda_{(r,0,0)}^\cT \right)^* u = 0,
\end{equation}
and reduce the general solution to only the square integrable ones. Finally, we use these solutions to compute the Higgs field.

\smallskip

\begin{remark}
	Before we begin, let us make a few important remarks:
	\begin{enumerate}
		\item We only carry out the computation in the cases of found in \Cref{theorem:3+1Monopoles} and in the $n = 2$ case\footnote{Explicit solutions for any $n \in \N_+$ can be found similarly.} of \Cref{theorem:(n+2)+nMonopoles}. While we also found an explicit, irreducible solution to Nahm's equations in \Cref{theorem:5+3+1Monopoles}, we could not perform the Nahm transform of it with our current technique. Nonetheless, general theory (cf. \cite{Charbonneau-Nagy-NahmTransform}) tells us that there is monopole with nonmaximal symmetry breaking corresponding to this Nahm datum, whose Higgs field, in some gauge, has the following asymptotic expansion
		\begin{equation}
			\Phi (r) = i \cdot \diag (1, 1, 1, - 1, - 1, - 1) - \frac{i}{2r} \diag (3, 3, 3, - 3, - 3, - 3) + O \left( r^{-2} \right).
		\end{equation}
		In more geometric terms, this monopole (topologically) decomposes the trivial $\SU (6)$ bundle over the ``sphere at infinity'' into two, nontrivial $\rU (3)$ bundles, both of which are further decomposed (holomorphically) into line bundles with Chern numbers $\pm 3$.

		\item By the second bullet point in \Cref{remark:spherical_reducibility}, the Nahm data found in \Cref{theorem:3+1Monopoles,theorem:5+3+1Monopoles,theorem:(n+2)+nMonopoles} are all irreducible, hence so are the corresponding BPS monopoles.

		\item For the computations below, one needs to know explicit formulae for $\left( Y_1^n, Y_2^n, Y_3^n \right)$ and $\left( B_1^n, B_2^n, B_3^n \right)$. The former are well-known in the literature, while the latter can be computed uniquely (up to a $\rU (1)$ factor) via \cref{eq:B_i_def,eq:B_norm_1}.
	\end{enumerate}
\end{remark}

\medskip

\subsubsection*{The $\left( V, \rho \right) \simeq \left( V_3, \rho_3 \right) \oplus \left( V_1, \rho_1 \right)$ case:} We have the following Nahm data from \Cref{theorem:3+1Monopoles}\footnote{We use a slightly different (but gauge equivalent) version of the solutions found in \Cref{theorem:3+1Monopoles}.}:
\begin{align}
	T_1	&=	\begin{bmatrix}
				0 & 0 & 0 & g\\
				0 & 0 & f & 0\\
				0 & - f & 0 & 0\\
				g & 0 & 0 & 0
			\end{bmatrix}, \quad
	T_2 =	\begin{bmatrix}
				0 & 0 & f & 0\\
				0 & 0 & 0 & - g\\
				- f & 0 & 0 & 0\\
				0 & - g & 0 & 0
			\end{bmatrix}, \quad
	T_3 =	\begin{bmatrix}
				0 & f & 0 & 0\\
				- f & 0 & 0 & 0\\
				0 & 0 & 0 & g\\
				0 & 0 & g & 0
			\end{bmatrix}, \\
	f &= - \frac{t}{t^2 - 1}, \quad g = \frac{i}{t^2-1}.
\end{align}
After the Nahm transform we get a Higgs field of the form
\begin{equation}
	\Phi (r) = i \cdot \diag (F (r), G (r), - F (r), - G (r)),
\end{equation}
where the functions $F$ and $G$ can be found in \cref{eq:FG_3+1} in \Cref{app:3+1}. The asymptotic expansion of this Higgs field is
\begin{equation}
	\Phi (r) = i \cdot \diag (1, 1, - 1, - 1) - \frac{i}{2r} \diag (2, 2, - 2, - 2) + O \left( r^{-2} \right).
\end{equation}
In more geometric terms, this monopole (topologically) decomposes the trivial $\SU (4)$ bundle over the ``sphere at infinity'' into two, nontrivial $\rU (2)$ bundles, both of which are further decomposed (holomorphically) into line bundles with Chern numbers $\pm 2$.

\Cref{fig:3+1} shows the pointwise norm squared of this Higgs field and energy density (as a function of $r$) of this $\SU (4)$ monopole. Note that $|\Phi|^2$ has a single zero at the origin.

\begin{figure}[H]
	\centering
	\includegraphics[width=\textwidth]{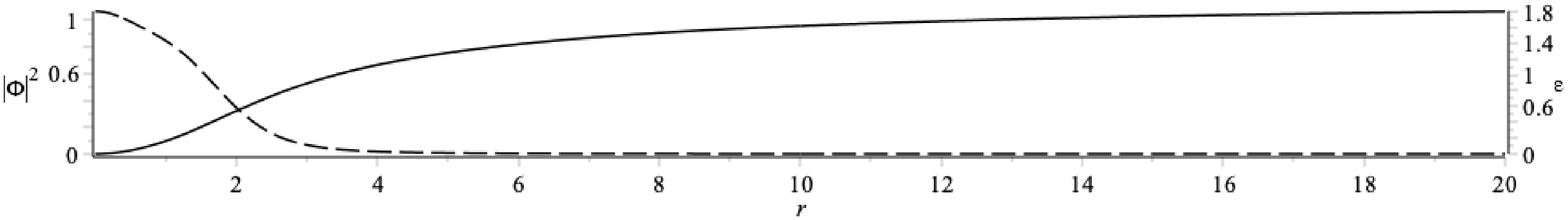}
	\caption{Norm squared of the Higgs field (solid) and energy density (dashed) of the spherically symmetric monopole corresponding to the solution to Nahm's equations found in \Cref{theorem:3+1Monopoles}. \\
	The representation given by the generators of the spherical symmetry satisfies $\left( V, \rho \right) \simeq \left( V_3, \rho_3 \right) \oplus \left( V_1, \rho_1 \right)$.}
	\label{fig:3+1}
\end{figure}

\smallskip

\subsubsection*{The $\left( V, \rho \right) \simeq \left( V_4, \rho_4 \right) \oplus \left( V_2, \rho_2 \right)$ case:} We the following Nahm data from \Cref{theorem:(n+2)+nMonopoles}:
\begin{align}
	T_1	&=	\begin{bmatrix}
				\frac{3i}{2}f^+ & 0 & 0 & 0 & 0 & 0\\
				0 & \frac{i}{2}f^+ & 0 & 0 & \sqrt{\frac{2}{3}}g & 0\\
				0 & 0 & -\frac{i}{2}f^+ & 0 & 0 & \sqrt{\frac{2}{3}}g\\
				0 & 0 & 0 & -\frac{3i}{2}f^+ & 0 & 0\\
				0 & -\sqrt{\frac{2}{3}}g & 0 & 0 & \frac{i}{2}f^- & 0\\
				0 & 0 & -\sqrt{\frac{2}{3}}g & 0 & 0 & -\frac{i}{2}f^-
			\end{bmatrix}, \\
	T_2	&=	\begin{bmatrix}
				0 & \frac{\sqrt{3}}{2}f^+ & 0 & 0 & \frac{i}{\sqrt{2}}g & 0\\
				-\frac{\sqrt{3}}{2}f^+ & 0 & f^+ & 0 & 0 & \frac{i}{\sqrt{6}}g\\
				0 & -f^+ & 0 & \frac{\sqrt{3}}{2}f^+ & \frac{i}{\sqrt{6}}g & 0\\
				0 & 0 & -\frac{\sqrt{3}}{2}f^+ & 0 & 0 & \frac{i}{\sqrt{2}}g\\
				\frac{i}{\sqrt{2}}g & 0 & \frac{i}{\sqrt{6}}g & 0 & 0 & \frac{1}{2}f^-\\
				0 & \frac{i}{\sqrt{6}}g & 0 & \frac{i}{\sqrt{2}}g & -\frac{1}{2}f^- & 0
			\end{bmatrix}, \\
	T_3	&=	\begin{bmatrix}
				0 & \frac{i\sqrt{3}}{2}f^+ & 0 & 0 & -\frac{1}{\sqrt{2}}g & 0\\
				\frac{i\sqrt{3}}{2}f^+ & 0 & if^+ & 0 & 0 & -\frac{1}{\sqrt{6}}g\\
				0 & if^+ & 0 & \frac{i\sqrt{3}}{2}f^+ & \frac{1}{\sqrt{6}}g & 0\\
				0 & 0 & \frac{i\sqrt{3}}{2}f^+ & 0 & 0 & \frac{1}{\sqrt{2}}g\\
				\frac{1}{\sqrt{2}}g & 0 & -\frac{1}{\sqrt{6}}g & 0 & 0 & \frac{i}{2}f^-\\
				0 & \frac{1}{\sqrt{6}}g & 0 & -\frac{1}{\sqrt{2}}g & \frac{i}{2}f^- & 0
			\end{bmatrix}, \\
	f^+	&= - \frac{3t+1}{3(t^2-1)}, \quad
	f^-	= - \frac{3t+5}{3(t^2-1)}, \quad
	g	= \frac{2}{\sqrt{3}(t^2 - 1)}.
\end{align}
After the Nahm transform we get a Higgs field of the form
\begin{equation}
	\Phi_0 (r) = i \cdot \diag (F (r), G (r), 1 - (F (r) + F (- r) + G (r) + G (- r)), F (- r), G (- r)),
\end{equation}
where the function $F$ and $G$ can be found in \cref{eq:FG_4+2} in \Cref{app:4+2}. While this Higgs field is not traceless, and hence structure group of the monopole is only $\rU (5)$, the traceless part of $\Phi_0$, call $\Phi$, together with the same connection forms an $\SU (5)$ monopole. The asymptotic expansion of this $\SU (5)$ Higgs field is
\begin{equation}
	\Phi (r) = i \cdot \diag \left( \frac{4}{5}, \frac{4}{5}, \frac{4}{5}, - \frac{6}{5}, - \frac{6}{5} \right) - \frac{i}{2r} \cdot \diag (2, 2, 2, - 3, - 3) + O \left( r^{-2} \right).
\end{equation}
In more geometric terms, this monopole (topologically) decomposes the trivial $\SU (5)$ bundle over the ``sphere at infinity'' into a direct sum of a $\rU (3)$ and a $\rU (2)$ bundle, both of which are further decomposed (holomorphically) into line bundles with Chern numbers 2 and $- 3$, respectively.

\Cref{fig:4+2} shows the pointwise norm squared of this Higgs field and energy density (as a function of $r$) of this $\SU (5)$ monopole. Note that $|\Phi|^2$ has a single zero at the origin.

\begin{figure}[H]
	\centering
	\includegraphics[width=\textwidth]{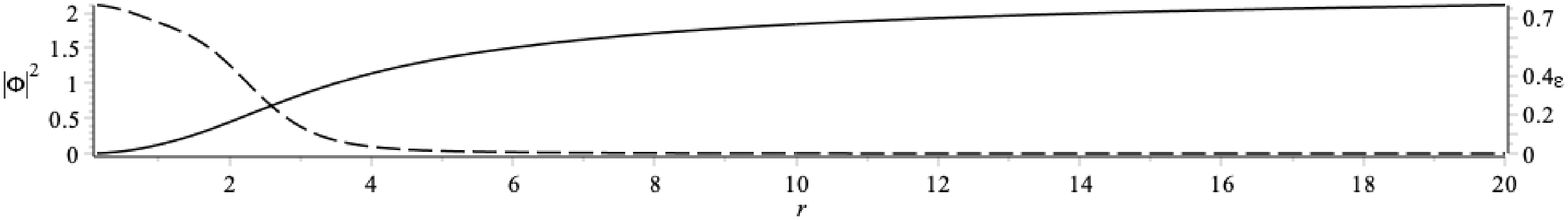}
	\caption{Norm squared of the Higgs field (solid) and energy density (dashed) of the spherically symmetric monopole corresponding to the solution to Nahm's equations found in \Cref{theorem:(n+2)+nMonopoles}. \\
	The representation given by the generators of the spherical symmetry satisfies $\left( V, \rho \right) \simeq \left( V_4, \rho_4 \right) \oplus \left( V_2, \rho_2 \right)$.}
	\label{fig:4+2}
\end{figure}

\bigskip

\begin{acknowledgment}
	BC acknowledges the support of the Natural Sciences and Engineering Research Council of Canada (NSERC), RGPIN-2019-04375.

	Portion of this work was done while CJL received NSERC USRA support. CJL was supported by NSERC CGS-D.

	AD, \'AN, and HY thank the organizers of DOmath 2019 at Duke University, where a large portion of this work was completed. In particular, they thank Matthew Beckett for his help during the program.
\end{acknowledgment}

\bigskip

\begin{data}
	The data that support the findings of this study are openly available on \href{https://github.com}{github.com}, at \cite{Maplecode}.
\end{data}

\bigskip

\appendix
\section{The coefficients of the Higgs fields}

\subsection{The $\left( V, \rho \right) \simeq \left( V_3, \rho_3 \right) \oplus \left( V_1, \rho_1 \right)$ case:}
\label{app:3+1}

\begin{equation}
\begin{aligned}\label{eq:FG_3+1}
	F (r)	&= \frac{e^{4 r} p_1 (r) - p_1 (- r)}{e^{4 r} p_2 (r) + p_2 (- r)}, \\
	G (r)	&= \frac{e^{12 r} p_3 (r) + e^{8 r} p_4 (r) - e^{4 r} p_4 (- r) - p_3 (- r)}{- e^{12 r} p_2 (r) + e^{8 r} p_5 (r) + e^{4 r} p_5 (- r) - p_2 (- r)}. \\
	p_1 (r)	&= 4r^2-6r+3, \\
	p_2 (r)	&= 4r^2-2r, \\
	p_3 (r)	&= 4r^2-6r+1, \\
	p_4 (r)	&= 64r^4-32r^3+36r^2+6r-3, \\
	p_5 (r)	&= 64r^4-32r^3+4r^2-6r.
\end{aligned}
\end{equation}

\smallskip

\subsection{The $\left( V, \rho \right) \simeq \left( V_4, \rho_4 \right) \oplus \left( V_2, \rho_2 \right)$ case:}
\label{app:4+2}

\begin{equation}
\begin{aligned}\label{eq:FG_4+2}
	F (r)		&= \frac{e^{4r}p_1(r)+p_2(r)}{e^{4r}p_3(r)-p_4(r)}, \\
	G (r)		&= \frac{e^{12r}p_5(r)-e^{8r}p_6(r)-e^{4r}p_7(r)+p_8(r)}{e^{12r}p_9(r)+e^{8r}p_{10}(r)-e^{4r}p_{11}(r)-p_4(r)}. \\
	p_1 (r)		&= 8r^3-16r^2+15r-6, \\
	p_2 (r)		&= 4r^2+9r+6, \\
	p_3 (r)		&= 8r^3-8r^2+3r, \\
	p_4 (r)		&= 4r^2+3r, \\
	p_5 (r)		&= 64r^5-192r^4+224r^3-112r^2+27r-3, \\
	p_6 (r)		&= 256r^6-64r^5+256r^4-80r^3-84r^2+45r-9, \\
	p_7 (r)		&= 128r^5+128r^4+112r^3-24r^2-9r+9, \\
	p_8 (r)		&= 4r^2+9r+3, \\
	p_9 (r)		&= 64r^5-128r^4+96r^3-32r^2+3r, \\
	p_{10} (r)	&= 256r^6-192r^5+64r^4-48r^3+60r^2-9r, \\
	p_{11} (r)	&= 128r^5+128r^4+48r^3+24r^2-9r.
\end{aligned}
\end{equation}

\bibliography{Bib-monopoles,references}
\bibliographystyle{abstract}

\end{document}